\documentclass[letterpaper,11pt]{article}
\usepackage[margin=1in]{geometry}
\usepackage{setspace}

\usepackage{enumitem}
\newlist{todolist}{itemize}{2}
\setlist[todolist]{label=$\square$}

\usepackage{xspace}
\usepackage{amssymb}
\usepackage{amsmath}
\usepackage{amsthm}
\usepackage{graphicx}
\usepackage{multicol}
\usepackage[colorlinks=true, allcolors=blue]{hyperref}
\usepackage{cleveref}
\usepackage[ruled,linesnumbered,vlined]{algorithm2e}
\usepackage{xcolor}
\usepackage{tcolorbox}
\usepackage{multirow}

\usepackage{tikz}
\usepackage{pgfplots}
\pgfplotsset{compat=1.18}
\usepackage{framed}
\definecolor{shadecolor}{named}{lightgray}
\usepackage{bm}
\usepackage{subfigure}
\usepackage{booktabs}
\usepackage[para]{threeparttable} 
\usepackage{tcolorbox}
\tcbset{
  myframe/.style={
    colback=gray!5,
    colframe=black,
    boxrule=0.8pt,
    arc=4pt,
    outer arc=4pt,
    boxsep=5pt,
    left=5pt,
    right=5pt,
    top=5pt,
    bottom=5pt,
    enhanced,
    sharp corners
  }
}

\usetikzlibrary{arrows.meta} 
\usetikzlibrary{calc}
\usetikzlibrary{positioning} 
\usetikzlibrary{scopes}
\usetikzlibrary{shapes.geometric}
\usetikzlibrary{bending}

\newtheorem{theorem}{Theorem}[section]
\newtheorem{lemma}[theorem]{Lemma}

\newtheorem{fact}[theorem]{Fact}

\newtheorem{definition}[theorem]{Definition}

\newtheorem{example}[theorem]{Example}

\newcommand{\rank}{\mathrm{Rank}}
\newcommand{\share}{\mathrm{Share}}
\newcommand{\ALG}{\textsf{ALG}}
\newcommand{\OPT}{\textsf{OPT}}

\DeclareMathOperator{\E}{\mathbf{E}}
\def\d{\mathrm{d}}

\newcommand{\defeq}{\mathrel{\mathop:}=}

\newcommand{\algoname}{\textsc{Quadratic Ranking}\xspace}
\newcommand{\Real}{\mathbb{R}}

\allowdisplaybreaks[4]
\definecolor{myGreen}{RGB}{100,200,100}
\definecolor{myRed}{RGB}{200,100,100}
\definecolor{myBlue}{RGB}{30,100,200}

\newenvironment{proofof}[1]{{\vspace*{5pt} \noindent\bf Proof of #1:  }}{\hfill\rule{2mm}{2mm}\vspace*{5pt}}

\title{
    Edge-weighted Matching in the Dark
    \thanks{The codes of the computer-aided analyses in this paper are provided at \href{https://github.com/Jiahao566/Edge-weighted-Matching-in-the-Dark}{https://github.com/Jiahao566/Edge-weighted-Matching-in-the-Dark}.
    }
}
\author{
Zhiyi Huang \thanks{The University of Hong Kong. Email: zhiyi@cs.hku.hk.}
\and
Enze Sun \thanks{The University of Hong Kong. Email: sunenze@connect.hku.hk.}
\and
Xiaowei Wu\thanks{University of Macau. Email: xiaoweiwu@um.edu.mo. Xiaowei Wu is funded by the Science and Technology Development Fund (FDCT), Macau SAR (file no. 0147/2024/RIA2, 0014/2022/AFJ, 0085/2022/A, and 001/2024/SKL).}
\and
Jiahao Zhao \thanks{The University of Hong Kong. Email: zjiahao@connect.hku.hk.}
}
\date{July 2025}

\begin{document}

\begin{titlepage}
\thispagestyle{empty}
\maketitle

\begin{abstract}
\thispagestyle{empty}
We present a $0.659$-competitive Quadratic Ranking algorithm for the Oblivious Bipartite Matching problem, a distribution-free version of Query-Commit Matching.
This result breaks the $1-\frac{1}{e}$ barrier, addressing an open question raised by Tang, Wu, and Zhang (JACM 2023).
Moreover, the competitive ratio of this distribution-free algorithm improves the best existing $0.641$ ratio for Query-Commit Matching achieved by the distribution-dependent algorithm of Chen, Huang, Li, and Tang (SODA 2025).

Quadratic Ranking is a novel variant of the classic Ranking algorithm.
We parameterize the algorithm with two functions, and let two key expressions in the definition and analysis of the algorithm be quadratic forms of the two functions.
We show that the quadratic forms are the unique choices that satisfy a set of natural properties.
Further, they allow us to optimize the choice of the two functions using powerful quadratic programming solvers.

\end{abstract}

\end{titlepage}
\section{Introduction}

Motivated by diverse applications, including kidney exchange, gig economy, online advertising, and online dating markets, there has been a long line of research on finding high-value matching under the uncertainty of which edges exist.
Consider a set of objects to be matched, represented by the vertices of a graph.
We know the value of matching each pair of vertices, represented by a non-negative weight, but are unsure if the edge exists.
Such uncertainty arises for various reasons in different applications, such as the compatibility of patient-donor pairs in kidney exchange, the interview results in gig economy, etc.
In these applications, we can resolve the uncertainty about an edge through a query, but we must commit to matching the edge if it exists. 
For instance, a digital labor platform may recommend a freelancer to an employer, but cannot stop the resulting employment.
It is folklore that querying the edges in descending order of the weights yields at least half the total weight of the optimal matching.
Are there better algorithms?

\emph{Query-Commit Matching} is a popular stochastic model of edge existence uncertainty.
In this model, each edge $e$ exists independently with a \emph{known} probability $p_e$.
We measure an online algorithm by its \emph{competitive ratio}---the ratio between the expected weight of its matching and that of the optimal matching, where the expectations are over the randomness of the algorithm and the edge realizations.
Much progress has been made in the Query-Commit Matching model.
For bipartite graphs, \cite{GKS19} introduced a $(1-\frac{1}{e})$-competitive algorithm, breaking the $0.5$ barrier for the first time.
Later, \cite{DF23} gave an improved $0.6335$-competitive algorithm, while \cite{CHLT25} proposed the state-of-the-art $0.641$-competitive algorithm.
For general graphs, the best existing result is the $0.533$-competitive algorithm by \cite{FTWWZ21}.

The stochastic model suffers from two drawbacks in some applications.
First, it may be challenging to estimate the probabilities $p_e$, e.g., as a startup digital labor platform with limited data. 
Further, the assumption of independent edge realization may not be realistic, e.g., a freelancer who fails an interview may be less likely to pass the next one.

Alternatively, the literature has considered the \emph{Oblivious Matching} model in which the online algorithm has no prior knowledge about the underlying graph.
In other words, there are no probabilities $p_e$, and an online algorithm needs to work well against the worst-case graph.
Oblivious Matching is a more robust model---the competitive ratio of any online algorithm in this worst-case model directly applies to the stochastic model with \emph{any edge realization process}, even those with correlation among the edges.
However, this model is also more challenging and has few positive results thus far.
\cite{TWZ23} broke the $0.5$ barrier with their $0.501$-competitive algorithm for general graphs.
Further, they gave a simple $(1-\frac{1}{e})$-competitive algorithm for bipartite graphs, and raised an open question whether there is an algorithm with a competitive ratio strictly better than $1-\frac{1}{e}$.

\subsection{Our Contributions}

This paper answers \cite{TWZ23}'s open question about Oblivious Bipartite Matching affirmatively.
As a corollary, we also substantially improve the state-of-the-art ratio of Query-Commit Matching.

\begin{tcolorbox}[colback = gray!15, colframe = gray!15]
{\bf Main Theorem.~}
There is a polynomial-time $0.659$-competitive online algorithm for Oblivious Bipartite Matching.\\[2ex]
{\bf Corollary.~}
There is a polynomial-time $0.659$-competitive online algorithm for Query-Commit Bipartite Matching, even if the realizations of edges are arbitrarily correlated. 
\end{tcolorbox}

\clearpage

We obtain these results by proposing a new randomized primal-dual algorithm called \algoname.
It is parameterized by a non-increasing function $g$ and a non-decreasing function $h$ that jointly decide the primal and dual assignments as follows.
See \Cref{sec:quadratic-ranking} for the formal definition.

\medskip

\begin{tcolorbox}[title=\algoname]
\begin{itemize}[itemsep=0pt]
        \item For each vertex $v$, draw its \emph{rank} $y_v \in [0,1)$ uniformly and independently.
        \item Query the edges $(u,v)$ in descending order of the \emph{perturbed weights} $g(y_u)\cdot g(y_v)\cdot w_{uv}$.
        \item If $(u,v)$ is matched, let the dual variables for $u$ and $v$, referred to as $u$ and $v$'s gains,\\ be $\alpha_u = h(y_u) \cdot g(y_v) \cdot w_{uv}$ and  $\alpha_v = h(y_v) \cdot g(y_u) \cdot w_{uv}$.
    \end{itemize}
\end{tcolorbox}

\medskip

Randomized primal-dual algorithms have been widely studied with great success in different matching models, such as the classic Online Bipartite Matching \cite{DJK13}, its random-order variant \cite{HTWZ19,JW21,PT25} and edge-weighted variant~\cite{FHTZ22,BC21}, Fully Online Matching \cite{HKTWZZ20,HTWZ20,TZ24}, Oblivious Matching \cite{TWZ23} and Streaming Submodular Matching~\cite{LW21}.
These existing algorithms define an edge's priority (e.g., the perturbed weight of edge-weighted algorithms) to be how much one of its vertices would gain should the edge be matched.
\algoname is conceptually interesting for unbinding the concepts of edge priority and vertices' gains.
They are now only partially coupled by the ranks of the two vertices, via two functions $g$ and $h$.
See \Cref{subsection:existing-algorithm} for a further discussion.

\algoname satisfies a set of natural properties. 
First, edge $(u,v)$'s perturbed weight is symmetric with respect to (w.r.t.) the vertices' ranks $y_u, y_v$ (Symmetry).
Further, from the viewpoint of any vertex, the order by which the algorithm queries its incident edges is consistent with its preference, i.e., the descending order of the gains (Rank-Share Consistency).
Finally, a vertex's preference over its incident edges is independent of its own rank (Preference Consistency).
These properties are axiomatic in the sense that our quadratic forms (in $g$ and $h$) of perturbed weights and vertices' gains are the unique choice satisfying all three properties. 
Moreover, the \algoname algorithm preserves several monotonic properties of the original Ranking algorithm by \cite{KVV90}, e.g., adding a vertex to one side of the bipartite graph weakly increases the gains of all vertices on the other side.
We present these properties of \algoname in
\Cref{subsection:notations,subsection:monotonicities,subsection:dual-variables}, and defer the proof of the axiomatic claim to \Cref{subsection:aximatic}.

The final step in designing \algoname is choosing functions $g$ and $h$.
We present in \Cref{section:maximal-choice-of-g-h} a family of ``maximal'' choices of $g$ and $h$, such that $\big( g(y), h(y) \big)$ are the points on an arc of the unit circle in the first quadrant, and the tangent line at the left endpoint of the arc.
Building on this, we design a pair of closed-form functions $g$ and $h$ and break the $1-\frac{1}{e}$ barrier with an analytical analysis.
Moreover, we optimize the choice of $g$ and $h$ by considering step functions and using solvers.
Thanks to the quadratic forms of perturbed weights and vertices' gains, we can use the powerful solvers for quadratic programming, even though the problem is NP-hard in general.
The numerical results confirm the optimality of our ``maximal'' families of $g$ and $h$.
We leave it for future research to prove this optimality, and even to derive the closed-forms of the optimal $g$ and $h$.
We present in \Cref{section:computer-aided-competitive-analysis} the computer-optimized $g$ and $h$ for the proof of our Main Theorem, and defer the analytics argument and its worse ratio to \Cref{app:analytical-competitive-analysis}.

Finally, we complement the competitive ratio of \algoname with two hardness results that upper bound the competitive ratios of any online algorithms. 
We show that no algorithm has a competitive ratio better than $0.7961$ for the (unweighted) Oblivious Bipartite Matching, improving the best previous bound of $0.898$ \cite{CTT12}.
We also consider general graphs and show an improved upper bound of $0.7583$, improving the best previous bound of $0.792$ \cite{GT12}.

\begin{tcolorbox}[colback = gray!15, colframe = gray!15]
{\bf Hardness Results.~}
For the (unweighted) oblivious matching problem, no algorithm can achieve a competitive ratio strictly larger than $0.7961$ for bipartite graphs and $0.7583$ for general graphs. 
\end{tcolorbox}

\begin{table}[ht]
\centering
\caption{Summary of prior works and our results (in bold) for Oblivious Matching}
\medskip
\renewcommand{\arraystretch}{1.2}
\begin{tabular}{c@{\hskip 0.25in}c@{\hskip 0.25in}c}
\toprule
& Lower Bound & Upper Bound \\ 
\midrule
Bipartite Graphs & $1 - 1/e$ \cite{TWZ23} $\to$ $\mathbf{0.659}$ & $0.898$ \cite{CTT12} $\to$ $\mathbf{0.7961}$ \\ 
General Graphs & $0.501$ and $0.531$ (unweighted) \cite{TWZ23} & $0.792$ \cite{GT12} $\to$ $\mathbf{0.7583}$ \\
\bottomrule
\end{tabular}
\end{table}

\subsection{Further Related Works}

\paragraph{Unweighted and Vertex-Weighted Matching.}
While this paper focuses on the most general edge-weighted matching problem, the literature has also considered the unweighted and vertex-weighted versions of Oblivious Matching and Query-Commit Matching.
It is well known that the unweighted and vertex-weighted Ranking algorithms and their competitive ratios in the random-order model of Online Bipartite Matching~\cite{KMT11,MY11,HTWZ19,JW21,PT25} directly apply to Oblivious Bipartite Matching and Query-Commit Matching; see Section~\ref{subsection:existing-algorithm} for a discussion. 
The best existing competitive ratios are $0.696$ for unweighted bipartite matching \cite{MY11} using a factor revealing LP approach, and $0.686$ for vertex-weighted (on only one side) bipartite matching \cite{PT25} with a randomized primal-dual approach.
\cite{CHS24} gave a $0.705$-competitive Query-Commit Matching algorithm for unweighted and vertex-weighted bipartite graphs, as a by-product of their stochastic online correlated selection technique originally designed for the Online Stochastic Matching problem.
For general graphs, \cite{CCWZ18} gave a $0.523$-competitive algorithm for unweighted matching.
\cite{CCW18} improved the ratio to $0.526$ and proposed a $0.501$-competitive algorithm for vertex-weighted matching.
These two results are based on the factor-revealing LP method.

\paragraph{Query-Commit with Patience Constraints.}
The original query-commit model by~\cite{CIKMR09} considers unweighted matching and involves \emph{patience constraints} that limit the maximum number of queries to each vertex's incident edges.
The state-of-the-art competitive ratios are $0.5$ for unweighted graphs \cite{A11}, $0.395$ for weighted bipartite graphs \cite{PRSW22}, and $0.269$ for weighted general graphs \cite{ABGSSX20}.

\paragraph{Randomized Greedy Matching.}
There is a literature started by \cite{DF91} on understanding randomized greedy algorithms for the (offline) maximum matching problem.
Although exact algorithms exist, researchers still hope to understand the randomized greedy algorithms for their simplicity and practical performance. 
\cite{ADFS95,PS12} showed that the \textit{Modified Randomized Greedy} algorithm guarantees a $(0.5 + \varepsilon)$-approximation for a small constant $\varepsilon$. 
\cite{TWZ23} proposed a $0.531$-approximation Random Decision Order algorithm.
Our algorithm can be seen as a randomized greedy algorithm for edge-weighted bipartite matching.

\section{Preliminary}

\paragraph{Model.}
Consider a weighted bipartite graph $G(L, R, E, w)$ where $L$ and $R$ are the sets of vertices on the left and the right, $E \subseteq L \times R$ is the set of edges, and $w = (w_{uv})_{u \in L, v \in R}$ specifies the non-negative weight of each edge (when it exists). 
Let $V = L \cup R$ denote the set of all vertices. 
The set of vertices and the edge weights are given to the algorithm at the beginning. 
The set of edges, however, is initially unknown to the algorithm. 
The algorithm iteratively and adaptively queries the existence of the edges in discrete time steps, to find a bipartite matching $M$ whose total weight is as large as possible. 
In each step, the algorithm chooses a pair of unmatched vertices $u \in L$ and $v \in R$, based on all information revealed so far, and queries whether edge $(u, v)$ exists.
We will slightly abuse notation and say that the algorithm queries edge $(u, v)$, even though edge $(u, v)$ may not exist.
If the queried edge exists, the algorithm must include it in the matching.

This model is known as \emph{Oblivious Bipartite Matching} and extends naturally to general graphs. 
We give a hardness result for the non-bipartite case in \Cref{ssec:hardness_general}, and leave for future research whether the approach of this paper could yield improved algorithms for non-bipartite graphs.

Given any instance and any online algorithm, we let $\OPT$ denote the total weight of the offline optimal matching, and let $\ALG$ denote the total weight of the algorithm's matching.
An algorithm is $F$-competitive if we have $\E [ \ALG ] \ge F \cdot \OPT$ for any instance.

\paragraph{Non-adaptive Algorithm.}
An algorithm is \emph{non-adaptive} if it chooses an order of the possible edges $L \times R$ at the beginning, and then queries them one at a time whenever both vertices are still unmatched.
The algorithm in this paper is non-adaptive, while our hardness results apply to adaptive algorithms.

\paragraph{Randomized Primal-Dual.}
The analysis in this paper builds on the randomized primal-dual technique by \cite{DJK13}. 
Consider the standard linear program relaxation of edge-weighted bipartite matching and its dual program.
\begin{align*}
	\max \quad & \textstyle \sum_{(u,v)\in E} ~ w_{uv}\cdot x_{uv} && \qquad\qquad & \min \quad & \textstyle\sum_{v \in V} \alpha_v \\
	\text{s.t.} \quad & \textstyle \sum_{v:(u,v)\in E} x_{uv} \leq 1 && \forall u\in L & \text{s.t.} \quad & \alpha_u + \alpha_v \geq w_{uv} && \forall (u,v)\in E \\
	& \textstyle \sum_{u:(u,v)\in E} x_{uv} \leq 1 && \forall v\in R & & \alpha_v \geq 0 && \forall v \in V \\
	& x_{uv} \geq 0 && \forall (u,v)\in E %
\end{align*}

A primal-dual algorithm constructs a matching online, letting the primal variables $x_{uv}$ be the indicator of including edge $(u, v)$ in the matching, and at the same time maintains a set of non-negative dual variables. 
Specifically, whenever our algorithm matches a pair of vertices $u$ and $v$, it will split the increase of the primal objective, i.e., $w_{uv}$, between dual variables $\alpha_u$ and $\alpha_v$. 
Our algorithm is randomized, and thus, the above primal and dual assignments are random variables. 

\begin{lemma} \label{lemma:primal_dual}
    A primal-dual algorithm's competitive ratio is at least $F$ if it satisfies
    \begin{itemize}
        \item $\sum_{(u, v)\in E}w_{uv}\cdot x_{uv} = \sum_{v \in V} \alpha_v$; and
        \item $\E[\alpha_u + \alpha_v]\ge F\cdot w_{uv}$ for all edges $(u, v)$, over the randomness of the algorithm.
    \end{itemize}
\end{lemma}

\begin{proof}
    By the first property, we have $\E[\ALG] = \E[\sum_{(u, v)\in E}w_{uv}\cdot x_{uv}] = \E[\sum_{v\in V}\alpha_v]$.
    By the second property,~$\frac{1}{F} \E[ \alpha_v]$ for $v \in V$ is feasible for the dual.
    Hence, we have $\sum_{v\in V} \frac{1}{F} \E[\alpha_v] \ge \OPT$ by weak duality.
    Combining the two parts, we conclude that $\E[\ALG] \ge F \cdot \OPT$.
\end{proof}

\section{Quadratic Ranking}
\label{sec:quadratic-ranking}

The algorithm is parameterized by two functions $g, h : [0, 1] \to \Real_{\ge 0}$ satisfying that
\begin{itemize}
    \item $g$ is non-increasing and right-continuous, and satisfies $g(y) > 0$ for $y \in [0, 1)$ and $g(1) = 0$;
    \item $h$ is non-decreasing and right-continuous, and satisfies $h(y) > 0$ for $y \in [0, 1]$; and
    \item For any $x, y \in [0, 1]$, $h(x)\cdot g(y) + h(y)\cdot g(x) \leq 1$.
\end{itemize}

\paragraph{Primal Matching.}
The primal matching algorithm belongs to the family of \emph{perturbed greedy} algorithms.
In the beginning, draw a rank $y_v\in [0,1)$ uniformly at random and independently for every vertex $v \in V$.
Let the \emph{perturbed weight} of pair $(u,v)\in L\times R$ be
\begin{equation}
    \label{eqn:perturbed-weight}
    \hat{w}(u,v) ~\defeq~ g(y_u)\cdot g(y_v)\cdot w_{uv}.
\end{equation}
Then, query the edges in descending order of their perturbed weights.

By adding positive infinitesimal noises to the edge-weights, we may assume without loss of generality that the edge-weights are distinct, and so are the perturbed weights.
The perturbed weights without ties give a ranking of the edges---an edge $(u, v)$ \emph{ranks higher} than another edge $(u', v')$ if $\hat{w}(u, v) > \hat{w}(u', v')$.

\paragraph{Dual Gain Sharing.}

Whenever the algorithm matches an edge $(u,v)$ in the primal, it sets the dual variables of the two endpoints such that:
\begin{equation}
\label{eqn:gain-sharing}
\begin{aligned}
    \alpha_u & ~\ge~ h(y_u) \cdot g(y_v)\cdot w_{uv} ~, \\
    \alpha_v & ~\ge~ h(y_v) \cdot g(y_u) \cdot w_{uv} ~, \\
    \alpha_u + \alpha_v & ~=~ w_{uv} ~.
\end{aligned}
\end{equation}

This is feasible because of the constraint $h(x)\cdot g(y) + h(y)\cdot g(x) \leq 1$.
We will refer to the right-hand-side of the above inequalities as the \emph{guaranteed gains} of $u$ and $v$.

\bigskip

We name the algorithm \algoname because both the perturbed weights and the guaranteed gains are quadratic in the defining parameters, i.e., functions $g$ and $h$.
The quadratic forms lead to several structural properties to be shown in the rest of this section, and allow us to harness powerful solvers for quadratically constrained quadratic programs (QCQP) to optimize these parameters in the next section.

\bigskip

We conclude the subsection with a brief discussion about the assumed properties of functions $g$ and $h$.
As demonstrated above, $h(x)\cdot g(y) + h(y)\cdot g(x) \leq 1$ is driven by the gain-sharing rule.

Next, we explain how the monotonicities of $g$ and $h$ correspond to natural notions of edge priorities.
We follow the convention of the literature that smaller ranks mean higher priorities.%
\footnote{This convention may be inherited from the original Ranking algorithm by \cite{KVV90}, which considered discrete ranks with rank $1$ being the highest.}
By letting $g$ be non-increasing, a smaller rank $y_v$ implies a larger $g(y_v)$, and thus, larger perturbed weights for the edges incident to $v$.
From an economic viewpoint, a vertex $v$ needs to offer a larger proportion of the incident edge-weights to the neighbors to justify a higher priority.
The monotonicity of the functions is consistent with this intuition.
With $g$ being non-increasing and $h$ being non-decreasing, a smaller $y_v$ implies a smaller proportion $h(y_v) \cdot g(y_u)$ for $v$, and a larger proportion $h(y_u) \cdot g(y_v)$ for its neighbor $u$.

The assumption that $g, h$ are positive for any realizable rank $y \in [0, 1)$ is for the convenience of our arguments, as it prevents zero perturbed weights and guaranteed gains.
It can be removed with appropriate treatments of the boundary cases with zeros.

Finally, we let $g$ be right-continuous and have boundary condition $g(1) = 0$ in order to develop properties for the strictly decreasing $g$ in our analytical argument (\Cref{app:analytical-competitive-analysis}) and the step-function $g$ in the computer-aided analysis (\Cref{section:computer-aided-competitive-analysis}) under a unified framework.

\subsection{Terminologies and Notations}
\label{subsection:notations}

Our primal-dual assignments satisfy the first condition of \Cref{lemma:primal_dual} because $\alpha_u + \alpha_v = w_{uv}$ whenever the algorithm matches an edge $(u, v)$.
The focal point of our analysis is then to prove the second condition of \Cref{lemma:primal_dual}.
Following the existing randomized primal-dual analysis for matching, we will consider any edge $(u, v) \in E$ and analyze the expectation of $\alpha_u + \alpha_v$ \emph{over the randomness of $y_u$ and $y_v$, conditioned on any realization of the other vertices' ranks}.
We define the following notation to denote the matching for different values of $y_u$ and $y_v$.

\begin{definition}%
    \label{definition:matching}
    Let $M(y_u, y_v)$ denote the matching given by the algorithm when the ranks of $u$ and $v$ are $y_u$ and $y_v$.
    Define an artificial rank $\bot$ to indicate having the vertex removed, e.g., $M(\bot, y_v)$ is the matching when $u$ is removed and $v$'s rank is $y_v$.
\end{definition}

Next, we define the preference a vertex $v \in V$ has over its incident edges, which naturally yields a preference over different matchings.
There are two natural candidate preferences---larger perturbed weights and larger guaranteed gains.
These two preferences are the same thanks to our quadratic forms!
For example, if $v$ is on the right and $u, u'$ are two neighbors on the left, we have
\[
    \underbrace{\vphantom{\Big|} \hat{w}(u, v) \ge \hat{w}(u', v)}_{\text{pertrubted weights}}
    \quad\Leftrightarrow\quad
    g(y_u) \cdot w_{uv} \ge g(y_{u'}) \cdot w_{u'v} 
    \quad\Leftrightarrow\quad
    \underbrace{\vphantom{\Big|} g(y_u) \cdot h(y_v) \cdot w_{uv} \ge g(y_{u'}) \cdot h(y_v) \cdot w_{u'v}}_{\text{guaranteed gains}}
    ~.
\]

Since both preferences are identical to the descending order of $g(y_u) \cdot w_{uv}$, they are independent of $v$'s rank $y_v$.
We stress that the consistency of the two preferences and their independence in $v$'s rank are consequences of our quadratic perturbed weights and guaranteed gains.

\begin{definition}%
    \label{definition:better-matching}
    The preferences of $v \in R$ (and $u \in L$ by symmetry) are as follows.
    \begin{itemize}
        \item It prefers edge $(u, v)$ over edge $(u', v)$ if $g(y_u) \cdot w_{uv} \ge g(y_{u'}) \cdot w_{u'v}$.
        \item It prefers $M(y_u, y_v)$ over $M(y'_u, y'_v)$ if it is not matched in the latter, or it is matched in both but prefers the edge it matches to in the former in the above sense.
        \item It prefers edge $(u, v)$ (with ranks $y_u, y_v$) over $M(y'_u, y'_v)$ if it prefers the former over the edge it matches to (if any) in $M(y'_u, y'_v)$.
    \end{itemize}
\end{definition}

For example, we can use the above terminology to concisely state the next lemma, which asserts that $u$ and $v$ prefer having each other in the graph.

\begin{lemma}\label{lemma:add-neighbor}
    For any ranks $y_u$ and $y_v$, vertex $u$ prefers $M(y_u, y_v)$ over $M(y_u, \bot)$.
    Symmetrically, vertex $v$ prefers $M(y_u, y_v)$ over $M(\bot, y_v)$.
\end{lemma}

The lemma follows by a standard alternating-path argument (see, e.g., \cite{CCWZ18,CCW18,TWZ23} for some examples in Oblivious Matching).
In fact, it does not rely on our quadratic form of perturbed weights, and remains true if we query the edges greedily by an arbitrary order, for the corresponding definition of vertex preference.
We defer this standard argument to \Cref{app:proof4lemma:add-neighbor}.

Next, we take a closer look at the above matchings $M(y_u, y_v)$, $M(y_u, \bot)$, and $M(\bot, y_v)$ to define the concept of marginal ranks.

\begin{definition}[Marginal Rank]
    \label{definition:marginal-rank-informal}
    For any rank $y_v$ of vertex $v$, the \emph{marginal rank} $\theta(y_v)$ of vertex $u$ is the smallest threshold such that $v$ prefers edge $M(\bot, y_v)$ over $(u, v)$ when $y_u \geq \theta(y_v)$.
    For any rank $y_u$ of $u$, define the marginal rank $\beta(y_u)$ of $v$ symmetrically.
\end{definition}

Suppose $v$ is matched to $p$ in $M(\bot, y_v)$, with $p$ being a dummy vertex with $w_{pv} = 0$ (and any $y_p \in [0, 1]$) if $v$ is unmatched.
By definition, we have
\begin{equation}
    \label{eqn:marginal-rank}
    \theta(y_v) ~=~~ g^{-1}\left( \frac{g(y_p)\cdot w_{pv}}{w_{uv}} \right)
    ~.
\end{equation}
A similar expression holds for $\beta(y_u)$ by symmetry.
We will artificially extend the marginal ranks with $\theta(1) = \beta(1) = 1$ so that they are well defined on $[0, 1]$.

Here, we use the following notion of inverse functions for right-continuous monotone functions.
\Cref{fig:inverse-function} presents an illustration.

\begin{definition}[Inverse Function]
    \label{definition:inverse-function}
    Let $f: [0, 1]\to \mathbb R_{\ge 0}$ be a right-continuous function.
    Define $\inf\varnothing = 1$ since we consider subsets of $[0, 1]$. 
    If $f$ is non-increasing, its inverse is
    $$
    f^{-1}(y) ~\defeq~ \inf\big\{ x \in [0, 1]: f(x) \le y \big\}
    ~.
    $$
    If $f$ is non-decreasing, its inverse is
    $$
    f^{-1}(y) ~\defeq~ \inf\big\{ x \in [0, 1]: f(x) > y \big\}
    ~.
    $$
    Observe that $f^{-1}$ shares the same monotonicity as $f$, and is also right-continuous.    
\end{definition}

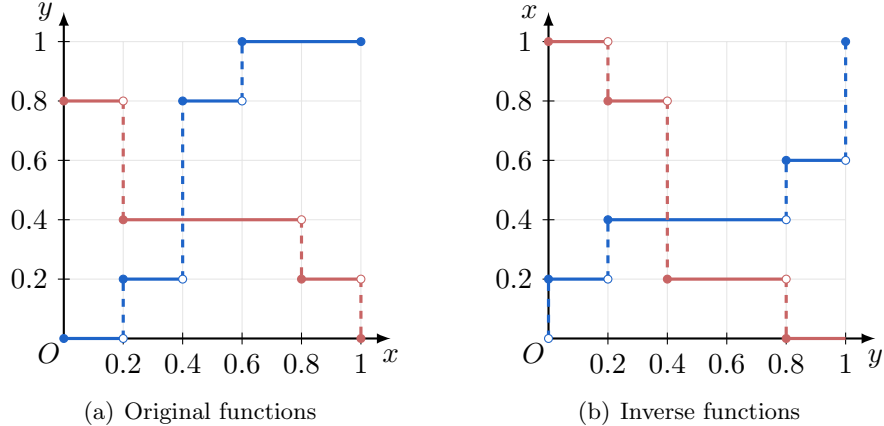
\begin{figure}[!ht]
    \centering
    \subfigure[Original functions]{
        \label{subfig:original-functions}
        \begin{tikzpicture}[scale=0.79]
        \draw[gray!20, thin] (0,0) grid (5,5);
        \draw[thick, -latex] (0,0) -- (5.5,0) node[below] {$x$};
        \draw[thick, -latex] (0,0) -- (0,5.5) node[left] {$y$};
        \node[black] at (-0.25, -0.25) {$O$};
    
        \foreach \x in {1,...,5}
            \draw (\x,0.1) -- (\x,-0.1) node[below] {$\pgfmathparse{\x/5}\pgfmathprintnumber[fixed, precision=1]{\pgfmathresult}$};
        \foreach \y in {1,...,5}
            \draw (0.1,\y) -- (-0.1,\y) node[left] {$\pgfmathparse{\y/5}\pgfmathprintnumber[fixed, precision=1]{\pgfmathresult}$};
    
        \draw[myBlue, very thick] (0,0) -- (1,0);
        \node[draw=myBlue, fill=myBlue, circle, minimum size=3pt, inner sep=0pt] at (0, 0) {};
        \draw[myBlue, very thick, dashed] (1, 0) -- (1,1); 
        \node[draw=myBlue, fill=white, circle, minimum size=3pt, inner sep=0pt] at (1, 0) {};
        \draw[myBlue, very thick] (1,1) -- (2,1); 
        \node[draw=myBlue, fill=myBlue, circle, minimum size=3pt, inner sep=0pt] at (1, 1) {};
        \draw[myBlue, very thick, dashed] (2,1) -- (2, 4); 
        \node[draw=myBlue, fill=white, circle, minimum size=3pt, inner sep=0pt] at (2, 1) {};
        \draw[myBlue, very thick] (2, 4)-- (3, 4); 
        \node[draw=myBlue, fill=myBlue, circle, minimum size=3pt, inner sep=0pt] at (2, 4) {};
        \draw[myBlue, very thick, dashed] (3, 4) -- (3, 5); 
        \node[draw=myBlue, fill=white, circle, minimum size=3pt, inner sep=0pt] at (3, 4) {};
        \draw[myBlue, very thick] (3, 5) -- (5, 5);
        \node[draw=myBlue, fill=myBlue, circle, minimum size=3pt, inner sep=0pt] at (3, 5) {};
        \node[draw=myBlue, fill=myBlue, circle, minimum size=3pt, inner sep=0pt] at (5, 5) {};
        
        \draw[myRed, very thick] (0,4) -- (1, 4) ;
        \node[draw=myRed, fill=myRed, circle, minimum size=3pt, inner sep=0pt] at (0, 4) {};
        \draw[myRed, very thick, dashed] (1, 4) -- (1, 2);
        \node[draw=myRed, fill=white, circle, minimum size=3pt, inner sep=0pt] at (1, 4) {};
        \draw[myRed, very thick] (1, 2)-- (4, 2);
        \node[draw=myRed, fill=myRed, circle, minimum size=3pt, inner sep=0pt] at (1, 2) {};
        \draw[myRed, very thick, dashed] (4, 2) -- (4, 1);
        \node[draw=myRed, fill=white, circle, minimum size=3pt, inner sep=0pt] at (4, 2) {};
        \draw[myRed, very thick] (4, 1) -- (5, 1);
        \node[draw=myRed, fill=myRed, circle, minimum size=3pt, inner sep=0pt] at (4, 1) {};
        \draw[myRed, very thick, dashed] (5, 1) -- (5, 0);
        \node[draw=myRed, fill=white, circle, minimum size=3pt, inner sep=0pt] at (5, 1) {};
        \node[draw=myRed, fill=myRed, circle, minimum size=3pt, inner sep=0pt] at (5, 0) {};

        \end{tikzpicture}
    }
    ~~~~
    \subfigure[Inverse functions]{
        \begin{tikzpicture}[scale=0.79]
        \draw[gray!20, thin] (0,0) grid (5,5);
        \draw[thick, -latex] (0,0) -- (5.5,0) node[below] {$y$};
        \draw[thick, -latex] (0,0) -- (0,5.5) node[left] {$x$};
        \node[black] at (-0.25, -0.25) {$O$};
    
        \foreach \x in {1,...,5}
            \draw (\x,0.1) -- (\x,-0.1) node[below] {$\pgfmathparse{\x/5}\pgfmathprintnumber[fixed, precision=1]{\pgfmathresult}$};
        \foreach \y in {1,...,5}
            \draw (0.1,\y) -- (-0.1,\y) node[left] {$\pgfmathparse{\y/5}\pgfmathprintnumber[fixed, precision=1]{\pgfmathresult}$};
    
        \draw[myBlue, very thick, dashed] (0, 0) -- (0,1); 
        \node[draw=myBlue, fill=white, circle, minimum size=3pt, inner sep=0pt] at (0, 0) {};
        \draw[myBlue, very thick] (0,1) -- (1,1);
        \node[draw=myBlue, fill=myBlue, circle, minimum size=3pt, inner sep=0pt] at (0, 1) {};
        \draw[myBlue, very thick, dashed] (1, 1) -- (1,2); 
        \node[draw=myBlue, fill=white, circle, minimum size=3pt, inner sep=0pt] at (1, 1) {};
        \draw[myBlue, very thick] (1,2) -- (4,2); 
        \node[draw=myBlue, fill=myBlue, circle, minimum size=3pt, inner sep=0pt] at (1, 2) {};
        \draw[myBlue, very thick, dashed] (4,2) -- (4, 3); 
        \node[draw=myBlue, fill=white, circle, minimum size=3pt, inner sep=0pt] at (4, 2) {};
        \draw[myBlue, very thick] (4, 3)-- (5, 3); 
        \node[draw=myBlue, fill=myBlue, circle, minimum size=3pt, inner sep=0pt] at (4, 3) {};
        \draw[myBlue, very thick, dashed] (5,3) -- (5, 5); 
        \node[draw=myBlue, fill=white, circle, minimum size=3pt, inner sep=0pt] at (5, 3) {};
        \node[draw=myBlue, fill=myBlue, circle, minimum size=3pt, inner sep=0pt] at (5, 5) {};
        
        \draw[myRed, very thick] (0,5) -- (1, 5) ;
        \node[draw=myRed, fill=myRed, circle, minimum size=3pt, inner sep=0pt] at (0, 5) {};
        \draw[myRed, very thick, dashed] (1, 5) -- (1, 4);
        \node[draw=myRed, fill=white, circle, minimum size=3pt, inner sep=0pt] at (1, 5) {};
        \draw[myRed, very thick] (1, 4)-- (2, 4);
        \node[draw=myRed, fill=myRed, circle, minimum size=3pt, inner sep=0pt] at (1, 4) {};
        \draw[myRed, very thick, dashed] (2, 4) -- (2, 1);
        \node[draw=myRed, fill=white, circle, minimum size=3pt, inner sep=0pt] at (2, 4) {};
        \draw[myRed, very thick] (2, 1) -- (4, 1);
        \node[draw=myRed, fill=myRed, circle, minimum size=3pt, inner sep=0pt] at (2, 1) {};
        \draw[myRed, very thick, dashed] (4, 1) -- (4, 0);
        \node[draw=myRed, fill=white, circle, minimum size=3pt, inner sep=0pt] at (4, 1) {};
        \draw[myRed, very thick] (4, 0) -- (5, 0);
        \node[draw=myRed, fill=myRed, circle, minimum size=3pt, inner sep=0pt] at (4, 0) {};
        \end{tikzpicture}
    }
    
    \caption{Inverse Function.}
    \label{fig:inverse-function}
\end{figure}

\begin{lemma}\label{lemma:right-continuous-effect}
    Consider a  non-increasing right-continuous function $f : [0, 1] \to \Real_{\ge 0}$ which satisfies $f(1) = 0$, we have $f \big( f^{-1}(y) \big) \le y$ for any $y \in [0, 1]$.
\end{lemma}
\begin{proof}
    If $f^{-1}(y) = 1$, then $f(f^{-1}(y)) = 0 \le y$.
    If $f^{-1}(y) < 1$, then $S = \{x\in[0, 1]: f(x)\le y\}$ is non-empty.
    Since $f$ is non-increasing and right-continuous, we have
    \begin{equation*}
        f \big( f^{-1}(y) \big) = \lim_{x \to f^{-1}(y)\,^+ } f(x) ~\le~ \sup_{x \in S} f(x) ~\le~ y
        ~. \qedhere
    \end{equation*}
\end{proof}

\subsection{Structural Monotonicities}
\label{subsection:monotonicities}

We first show that a vertex $v$ prefers the matching given by the algorithm when its rank $y_v$ is smaller, over the counterpart when its rank is larger.
In other words, a smaller rank leads to a more preferable matching from the vertex's point of view.

\begin{lemma}[Monotonicity of Matching Outcome] \label{lemma:rank_monotonicity}
    For any $y_u \in [0,1) \cup \{\bot\}$ and any $0 \le y_v < y'_v < 1$, vertex $v$ prefers $M(y_u, y_v)$ over $M(y_u, y'_v)$.
    A similar claim holds for $u$ by symmetry.
\end{lemma}
\begin{proof}
    With the ranks of vertices other than $v$ fixed (and vertex $u$ removed if $y_u = \bot$), consider two runs of the algorithm with $v$'s rank being $y_v$ and $y'_v$ respectively.
    We will process the edges in the two runs in descending order of their perturbed weights.
    An edge incident to $v$ has a larger perturbed weight in the former run than in the latter.
    Hence, its copy in the former run will be processed earlier.
    An edge not incident to $v$ has the same perturbed weights in both runs. 
    Hence, its two copies will be processed at the same time.

    If $v$ is unmatched in the second run, the lemma holds trivially.
    Otherwise, suppose $v$ is matched to $p$ in the second run.
    Consider the moment when we process $(p, v)$ in the first run, which is earlier than when we process it in the second run.
    In other words, both $p$ and $v$ are still unmatched in the second run at this moment.
    If $v$ is already matched in the first run, the lemma follows.
    Otherwise, $v$ is unmatched in both runs.
    Then, the sets of matched edges must be the same in the two runs since other vertices have the same ranks. 
    In particular, both $p$ and $f$ are still unmatched in the first run, and would be matched to each other now.
    Hence, the lemma still holds in this case.
\end{proof}

As a corollary, we show that the marginal ranks $\theta$ and $\beta$ are monotone.

\begin{lemma}[Monotonicity of Marginal Ranks] \label{lemma:theta_beta_increasing}
    Both $\theta(y)$ and $\beta(y)$ are non-decreasing.
\end{lemma}
\begin{proof}
By symmetry, we only need to prove the lemma for $\theta(y)$. 
The boundary value $\theta(1) = 1$ satisfies monotonicity.
Next, consider any $0 \le y_v < y'_v < 1$.
Let $p$ and $p'$ be the vertices $v$ is matched with in $M(\bot, y_v)$ and $M(\bot, y'_v)$ respectively;
they may be dummies with zero edge-weights if $v$ is unmatched.
By \Cref{lemma:rank_monotonicity}, we have $g(y_p)\cdot w_{p v} \ge g(y_{p'})\cdot w_{p' v}$.
Therefore, we get $\theta(y_v) \ge \theta(y'_v)$ from \Cref{eqn:marginal-rank} and the non-increasingness of $g^{-1}$.
\end{proof}

\subsection{Lower Bounds for the Dual Variables}
\label{subsection:dual-variables}

We first present the following lower bounds based on the marginal ranks using \Cref{lemma:add-neighbor}.

\begin{lemma}%
    \label{lemma:basic-gain}
    For all $y_u, y_v\in [0, 1)$, we have
    \[
        \alpha_u \geq h(y_u)\cdot  g(\beta(y_u))\cdot w_{uv}
        ~,\quad
        \alpha_v \geq h(y_v)\cdot   g(\theta(y_v))\cdot w_{uv}
        ~.
    \]
\end{lemma}
\begin{proof}
    By symmetry, we only need to prove the bound for $\alpha_u$. 
    By \Cref{lemma:add-neighbor}, vertex $u$ prefers matching $M(y_u, y_v)$ over $M(y_u, \bot)$.
    Hence, it suffices to bound $u$'s guaranteed gain in the latter.
    By \Cref{lemma:right-continuous-effect} and the definition that $\beta(y_u) = g^{-1} \big( \frac{g(y_q)\cdot w_{u q}}{w_{uv}} \big)$, we have $g(\beta(y_u)) \leq \frac{g(y_q) \cdot w_{u q}}{w_{uv}}$.
    Hence, the value of $\alpha_u$ w.r.t.\ matching $M(y_u, \bot)$ is at least
    \begin{equation*}
        h(y_u)\cdot g(y_q)\cdot w_{uq} 
        ~=~ 
        h(y_u)\cdot \frac{g(y_q)\cdot w_{uq}}{w_{uv}}\cdot w_{uv} 
        ~\geq~ 
        h(y_u)\cdot  g(\beta(y_u))\cdot w_{uv}
        ~. \qedhere
    \end{equation*}
\end{proof}

Next, we show that $u$ and $v$ will be matched if their ranks are smaller than the marginal ones.

\begin{lemma}%
    \label{lemma:extra-gain}
    If $y_u < \theta(y_v)$ and $y_v < \beta(y_u)$, then $u$ and $v$ are matched, and
    \[
        \alpha_u + \alpha_v = w_{uv}
        ~.
    \]
\end{lemma}

\Cref{lemma:extra-gain} is an improvement compared to directly applying \Cref{lemma:basic-gain} and then relaxing the marginal ranks with the given conditions, i.e.,
\[
    h(y_u)\cdot g(\beta(y_u))\cdot w_{uv} 
    + h(y_v)\cdot g(\theta(y_v))\cdot w_{uv}
    \le 
    (h(y_u)\cdot g(y_v) + h(y_v)\cdot g(y_u))\cdot w_{uv} 
    \leq 
    w_{uv}
    ~.
\]

\Cref{lemma:extra-gain} is a corollary of the next lemma.

\begin{lemma} %
    If $y_u < \theta(y_v)$, then the algorithm matches $u$ no later than when it matches $v$.
    The symmetric claim also holds.
\end{lemma}

\begin{proof}
    Recall that
    \begin{equation*}
        \theta(y_v) = g^{-1}\left(\frac{g(y_p) w_{pv}}{w_{uv}}\right)
        = \inf\left\{ x\in [0,1] : g(x)\leq \frac{g(y_p) w_{pv}}{w_{uv}} \right\},
    \end{equation*}
    where $p$ is the vertex $v$ is matched with in $M(\bot, y_v)$.
    Therefore, for $y_u < \theta(y_v)$ we have
    \begin{equation*}
        g(y_u)\cdot w_{uv} ~>~ g(y_p)\cdot w_{pv}.
    \end{equation*}

    Now consider the matching $M(y_u, y_v)$.
    The above inequality implies that $(u, v)$ ranks higher than $(p, v)$.
    Hence, $v$ is still unmatched when the algorithm checks edge $(u,v)$.
    At that moment, either $u$ is already matched (with a vertex other than $v$), or $u$ is now matched with $v$.
    The lemma holds in either case.
\end{proof}

With \Cref{lemma:basic-gain,lemma:extra-gain}, we are able to derive the lower bound for $\frac{\E[\alpha_u + \alpha_v]}{w_{uv}}$ in \Cref{lemma:universal-lower-bound}. 

\begin{lemma}[Universal Lower Bound]\label{lemma:universal-lower-bound}
    We have
    \begin{align*}
        \frac{\E[\alpha_u + \alpha_v]}{w_{uv}} \geq \int_0^1 \underbrace{\Big( \theta(y_v) - \beta^{-1}(y_v) \Big)^+}_{\text{\rm shaded area, \Cref{lemma:extra-gain}}} \d y_v & + \int_0^1 \underbrace{\Big(1 - \big(\theta(y_v) - \beta^{-1}(y_v)\big)^+ \Big) \cdot h(y_v) \cdot  g(\theta(y_v))}_{\text{\rm non-shaded area, \Cref{lemma:basic-gain} for $\alpha_v$}} \,\d y_v \\
        & + \int_0^1 \underbrace{\Big(1 - \big(\beta(y_u) - \theta^{-1}(y_u)\big)^+\Big)\cdot h(y_u) \cdot g(\beta(y_u))}_{\text{\rm non-shaded area, \Cref{lemma:basic-gain} for $\alpha_u$}}  \,\d y_u
        ~,
    \end{align*}
    where $x^+$ denotes $\max\{x, 0\}$.
\end{lemma}

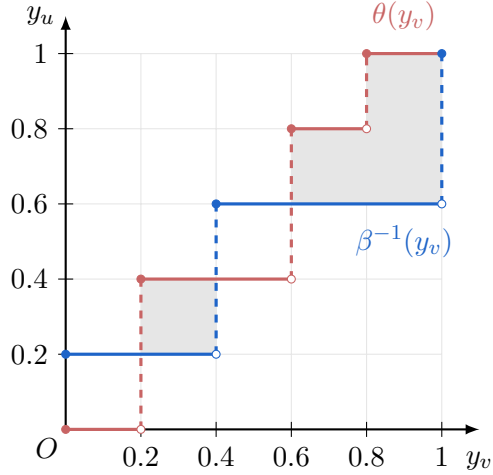
\begin{figure}[!ht]
\centering
\begin{tikzpicture}
    \draw[gray!20, thin] (0,0) grid (5,5);
    \draw[thick, -latex] (0,0) -- (5.5,0) node[below=0.15] {$y_v$};
    \draw[thick, -latex] (0,0) -- (0,5.5) node[left] {$y_u$};
    \node[black] at (-0.25, -0.25) {$O$};

    \foreach \x in {1,...,5}
        \draw (\x,0.1) -- (\x,-0.1) node[below] {$\pgfmathparse{\x/5}\pgfmathprintnumber[fixed, precision=1]{\pgfmathresult}$};
    \foreach \y in {1,...,5}
            \draw (0.1,\y) -- (-0.1,\y) node[left] {$\pgfmathparse{\y/5}\pgfmathprintnumber[fixed, precision=1]{\pgfmathresult}$};

    \draw[draw=none,fill=gray!20] (1,1) rectangle (2,2);
    \draw[draw=none,fill=gray!20] (3,3) rectangle (5,4);
    \draw[draw=none,fill=gray!20] (4,4) rectangle (5,5);
            
    \draw[myRed, very thick] (0,0) -- (1,0);
    \draw[myRed, very thick, dashed] (1,0) -- (1,2); 
    \draw[myRed, very thick] (1,2) -- (3,2); 
    \draw[myRed, very thick, dashed] (3,2) -- (3,4); 
    \draw[myRed, very thick] (3,4)-- (4,4); 
    \draw[myRed, very thick, dashed] (4,4) -- (4,5);
    \draw[myRed, very thick] (4,5)-- (5,5); 
    \node[draw=myRed, fill=myRed, circle, minimum size=3pt, inner sep=0pt] at (0, 0) {};
    \node[draw=myRed, fill=myRed, circle, minimum size=3pt, inner sep=0pt] at (1, 2) {}; 
    \node[draw=myRed, fill=myRed, circle, minimum size=3pt, inner sep=0pt] at (3, 4) {};
    \node[draw=myRed, fill=myRed, circle, minimum size=3pt, inner sep=0pt] at (4, 5) {};
    \node[draw=myRed, fill=white, circle, minimum size=3pt, inner sep=0pt] at (1, 0) {};
    \node[draw=myRed, fill=white, circle, minimum size=3pt, inner sep=0pt] at (3, 2) {};
    \node[draw=myRed, fill=white, circle, minimum size=3pt, inner sep=0pt] at (4, 4) {};
    \node at (4.5,5.5) {\textcolor{myRed}{$\theta(y_v)$}};
    \draw[myBlue, very thick] (0,1) -- (2,1);
    \draw[myBlue, very thick, dashed] (2,1) -- (2,3); 
    \draw[myBlue, very thick] (2,3) -- (5,3); 
    \draw[myBlue, very thick, dashed] (5,3) -- (5,5);
    \node[draw=myBlue, fill=myBlue, circle, minimum size=3pt, inner sep=0pt] at (0,1) {};
    \node[draw=myBlue, fill=myBlue, circle, minimum size=3pt, inner sep=0pt] at (2,3) {};
    \node[draw=myBlue, fill=myBlue, circle, minimum size=3pt, inner sep=0pt] at (5,5) {};
    \node[draw=myBlue, fill=white, circle, minimum size=3pt, inner sep=0pt] at (2,1) {};
    \node[draw=myBlue, fill=white, circle, minimum size=3pt, inner sep=0pt] at (5,3) {};
    \node at (4.5,2.5) {\textcolor{myBlue}{$\beta^{-1}(y_v)$}};
\end{tikzpicture}
\caption{Universal Lower Bound}
\label{fig:universal_lower_bound}
\end{figure}

As in the example shown by \Cref{fig:universal_lower_bound}, the first term on the right-hand side corresponds to the improved gain of $\alpha_v + \alpha_v$ in the shaded area, and the last two terms correspond to the guaranteed gain of $\alpha_v$ and $\alpha_u$ in the non-shaded area, respectively.
Similar bounds for the dual variables also appeared in~\cite{HTWZ19,JW21,PT25} for vertex-weighted online bipartite matching with random arrivals. 
However, these previous works failed to extend their techniques to edge-weighted matching.
Our main technical contribution is the introduction of quadratic perturbed weights and guaranteed gains and their structural properties, so that we could prove the above lemmas for edge-weighted matching. 
We defer the detailed proof to \Cref{app-subsec:proof-universal-lower-bound}.

\subsection{Choosing Functions \texorpdfstring{$g$}{g} and \texorpdfstring{$h$}{h}}\label{section:maximal-choice-of-g-h}

We conclude this section with a discussion about how to choose functions $g$ and $h$.
In particular, let us focus on the most nontrivial constraint that for all $x, y\in [0, 1]$
\begin{equation}
    \label{eqn:budget-balanced}
    h(x)\cdot g(y) + h(y)\cdot g(x)\le 1
    ~.
\end{equation}

\begin{example}[Unit Circle]
    \label{example:unit-circle}
    If $g(y)^2 + h(y)^2 = 1$ for all $y \in [0, 1]$, Constraint \eqref{eqn:budget-balanced} follows from the fact that the dot product of any two unit vectors is at most $1$.
\end{example}

\begin{example}[Linear Segment]
    \label{example:line-segment}
    If $g(y) \le 1$ and $h(y) = 1 - \frac{1}{2} g(y)$ for all $y \in [0, 1]$, Constraint \eqref{eqn:budget-balanced} follows from
    \[
        \Big(1 - \frac{1}{2} g(x) \Big) g(y) + \Big(1 - \frac{1}{2} g(y) \Big) g(x) = g(x) + g(y) - g(x)g(y) \le 1
        ~,
    \]
    where the last inequality follows from $g(x), g(y) \le 1$.
\end{example}

The first example is easy to find by the dot-product-like left-hand-side of the constraint.
We derive the second example by assuming $g$ and $h$ to be linearly related and then solving for the coefficients of the linear relation.

Moreover, we find a general family of $g$ and $h$ parameterized by $0 \le \varphi \le \frac{\pi}{4}$.
Both examples above are special cases of this family of functions.

\begin{lemma}[General Form of $g$ and $h$]
\label{lemma:general-h-satisfy-constraint}
    For any function $g : [0, 1] \to [0, \cos \varphi]$, and 
    \[
        h(y) = \begin{cases}
                \dfrac{1}{\cos \varphi} - g(y) \cdot \tan \varphi & \mbox{if $0\le g(y)< \sin \varphi$\,;} \\[2ex]
                \sqrt{1 - g(y)^2} & \mbox{if $\sin \varphi\le g(y) \le \cos \varphi$\,,}
        \end{cases}
    \]
    Constraint \eqref{eqn:budget-balanced} holds for any $x, y \in [0, 1]$.
\end{lemma}

\Cref{fig:general-form-h} illustrates the general form of $g$ and $h$ and we defer the simple proof of the lemma to \Cref{app:proof-general-h-satisfy-constraint}.
\Cref{example:unit-circle} is the special case when $\varphi = 0$.
\Cref{example:line-segment} is the special case when $\varphi = \frac{\pi}{4}$, except that we scale $g$ by $\sqrt{2}$ and $h$ by $\frac{1}{\sqrt{2}}$, noting that the algorithm's behavior and the guaranteed gains are both invariant to such a scaling.

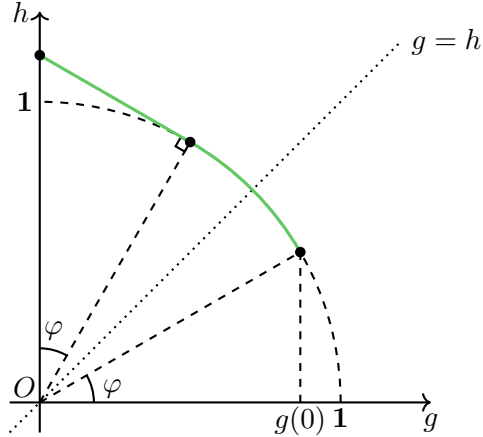
\begin{figure}[!ht]
    \centering
    \begin{tikzpicture}[scale = 4]
        \draw[->, thick] (-0.1,0) -- (1.3,0) node[below] {$g$};
        \draw[->, thick] (0,-0.1) -- (0,1.3) node[left] {$h$};

        \draw[black, thick, dashed] (1, 0) arc (0:90:1);

        \node at (-0.05, 1) {$\mathbf{1}$};
        \node at (1, -0.05) {$\mathbf{1}$};
        \node at (-0.05, 0.05) {$O$};
        \draw[-, thick, dashed] (0.866, 0.5) -- (0.866, 0);
        \node at (0.866, -0.06) {$g(0)$};

        \coordinate (A) at (0.5,0.866);
        \draw[black, thick, rotate=60] (A) rectangle ($(A) + (-0.035,0.035)$);
        
        \draw[myGreen, very thick] (0,1.155) -- (0.5,0.866);  

        \draw[black, dotted, thick] (-0.1, -0.1) -- (1.2, 1.2) node[right] {$g = h$};
        
        \draw[myGreen, very thick] (0.866, 0.5) arc (30:60:1);

        \fill[black] (0.5,0.866) circle (0.5pt);

        \fill[black] (0.866, 0.5) circle (0.5pt);

        \fill[black] (0, 1.155) circle (0.5pt);

        \draw[black, thick] (0,0) ++(0:0.18cm) arc (0:30:0.18cm) node[midway, right] {$\varphi$};  
        \draw[black, thick] (0,0) ++(60:0.18cm) arc (60:90:0.18cm) node[midway, above] {$\varphi$};

        \draw[black, thick, dashed]  (0, 0) -- (0.866, 0.5);
        \draw[black, thick, dashed]  (0, 0) -- (0.5,0.866);
    \end{tikzpicture}
    \caption{General form of $g$ and $h$ as points on the green curve}
    \label{fig:general-form-h}
\end{figure}

The above general form of $g$ and $h$ is \emph{maximal} in the sense that for any $x \in [0, 1]$, there is a corresponding $y \in [0, 1]$ for which Constraint \eqref{eqn:budget-balanced} holds with equality.
If $(g(x), h(x))$ is a point on the unit circle, $(g(y), h(y))$ is the reflection w.r.t.\ the line $g = h$.
If $(g(x), h(x))$ is on the tangent line, then $y = 0$ and $(g(y), h(y))$ is the right-most point of the green curve in \Cref{fig:general-form-h}.

\bigskip

Unfortunately, solving for the optimal functions $g$ and $h$ that maximize the lower bound in \Cref{lemma:universal-lower-bound} is analytically intractable. 
We explore two approaches to find approximately optimal $g$ and $h$ that yield non-trivial competitive ratios for the Oblivious Bipartite Matching problem.

\paragraph{Analytical Analysis.}
The first approach relaxes the lower bound in \Cref{lemma:universal-lower-bound} to get a sub-optimal yet analytically tractable bound. 
Following this approach, we derive a pair of closed-form $g$ and $h$ functions on the unit circle to show a competitive ratio slightly larger than the $1-\frac{1}{e}$ barrier.
Since the ratio is subsumed by the next approach, we defer its proof to \Cref{app:analytical-competitive-analysis}.

\begin{theorem}
    \label{theorem:analytical-lower-bound}
    For $g(y) = a - b \cdot \exp(\min\{y,c\})$, and $h(y) = \sqrt{1-g^2(y)}$, where $a = 1.171, b =0.339 , c = 0.652$, \algoname is $0.6324$-competitive.
\end{theorem}

\paragraph{Computer-Aided Analysis.}
We also consider step functions $g$ and $h$ and treat their values at each step as a finite number of variables. 
Then, we can employ powerful solvers for QCQP to optimize the step functions up to a certain number of steps.
Further, based on the optimal solution by QCQP solvers, we construct step functions $g$ and $h$ with even more steps and verify the resulting competitive ratios.
With this computer-aided approach, we get a much better competitive ratio than the analytical approach, even when we consider functions with less than $10$ steps.
The final bound is from a verification-based approach with $13$ steps.
We present the details in \Cref{section:computer-aided-competitive-analysis}.

\begin{theorem}[Main Theorem]
    \label{thm:main}
    There are step functions $g$ and $h$ such that \algoname is $0.659$-competitive for Oblivious Bipartite Matching.
\end{theorem}

\section{Computer-Aided Competitive Analysis}
\label{section:computer-aided-competitive-analysis}

This section focuses on a special class of functions $g$ and $h$, i.e., the step functions. 
The motivation is the tractability of the resulting optimization problem.
While the computational complexity is exponential in the number of steps $n$, we show that computer-assisted analyses are possible for small $n$, which already yields a much better competitive ratio than the analytical approach in \Cref{app:analytical-competitive-analysis}. 
This computer-aided approach is inspired by \cite{PT25}, who considered discretization with LP-based optimization for vertex-weighted online bipartite matching with random arrivals.

\subsection{Discretization}
\label{Subsection:discretization}

Let $n$ be a positive integer and define $[n] \defeq \{1, 2, \dots, n\}$. 
We consider functions $g$ and $h$ that are $n$-segment step functions defined as follows.

\begin{definition}[Step Function]
\label{definition:n-seg-step-function}

We say $f : [0, 1] \to \Real_{\ge 0}$ is an $n$-segment step function if there are constants $F_1, F_2, \dots, F_n \in \Real_{\ge 0}$ such that
\[
    f(x) = F_i~, \quad x\in \left[\frac{i - 1}{n}, \frac{i}{n}\right),\quad i\in [n]~.
\]
\end{definition}

By definition, $n$-segment step functions are right-continuous. 
We follow the convention that for any $n$-segment step function denoted by a lower-case letter, the corresponding upper-case letter represents the values of the segments.
For example, for an $n$-segment step function $g$, let
\[
    G_i = g\bigg(\frac{i - 1}{n}\bigg)~, i\in[n] ~.
\]

In line with previous sections, we consider non-increasing function $g$ with boundary value $g(1) = G_{n + 1} = 0$ in the subsequent discussion. Correspondingly, we also artificially extend the marginal ranks
with $\theta(1) = \Theta_{n + 1} =1 $ and $ \beta(1) = B_{n + 1} = 1$.

\begin{lemma}
    \label{lemma:step-marginal-ranks}
    If $g$ and $h$ are $n$-segment step functions, so are the marginal ranks $\theta$ and $\beta$.
    Further, the values of the marginal ranks are multiples of $\frac{1}{n}$ between $0$ and $1$ for any vertex rank $y \in [0, 1)$.
\end{lemma}

\begin{proof}
By symmetry, we only need to prove the lemma for $\theta$. 
Consider any $y_v \in [0, 1)$. 
Let $p$ be the vertex matched with $v$ in $M(\bot, y_v)$;
let $p$ be a dummy with zero edge-weights if $v$ is unmatched in $M(\bot, y_v)$. 
Recall that
\[
    \theta(y_v) 
    ~=~
    g^{-1}\left(\frac{g(y_p)\cdot w_{pv}}{w_{uv}}\right) 
    ~=~
    \inf\Big\{\, x \in [0, 1] : g(x) \le \frac{g(y_p)\cdot w_{pv}}{w_{uv}} \,\Big\}
    ~.
\]
Since $g$ is an $n$-segment step function, $\frac{g(y_p)\cdot w_{pv}}{w_{uv}}$ is an $n$-segment step function in $y_v$, and thus, so is $\theta$ by the above definition.
Further, the infimum always returns a multiple of $\frac{1}{n}$, because rounding $x$ down to the closest multiple of $\frac{1}{n}$ yields the same $g(x)$.
\end{proof}

\Cref{lemma:step-marginal-ranks} allows us to focus on a \emph{finite} (yet exponential) number of possible marginal ranks.
Since we will repeatedly refer to the set of possible marginal ranks, let us introduce a notation $\mathcal{S}_n$ to denote the set of non-decreasing $n$-segment step functions with range $\{0, \frac{1}{n}, \dots, \frac{n-1}{n}, 1\}$.
The conclusion of \Cref{lemma:step-marginal-ranks} can then be restated as $\theta, \beta \in \mathcal{S}_n$. 
By the definition of inverse functions (\Cref{definition:inverse-function}), the inverse threshold ranks $\theta^{-1}$ and $\beta^{-1}$ also belong to $\mathcal{S}_n$.
Observe that the size of $\mathcal S_n$ is equal to the number of $n$-element multi-set out of $n$ different elements. 
Therefore, we have $|\mathcal S_n| = \binom{2n - 1}{n} \le \sum_{i = 0}^{2n-1}\binom{2n - 1}{i} = 2^{2n - 1}$. Thus, we have the following fact.

\begin{fact}\label{fact:num-of-seg-step-function}
    $|\mathcal{S}_n| = O ( 4^n )$.
\end{fact}

Next, we substitute $n$-segment step functions $g$, $h$, and the corresponding threshold ranks to the lower bound in \Cref{lemma:universal-lower-bound} for general functions $g$ and $h$.

\begin{lemma}[Discretization Lower Bound]\label{lemma:discretization-lower-bound}
    Suppose $g$ and $h$ are $n$-segment step functions. 
    For any marginal ranks $\theta, \beta \in \mathcal{S}_n$, $\frac{\E[\alpha_u + \alpha_v]}{w_{uv}}$ is at least
    \begin{equation}
        \label{eqn:step-function-bound}
        \begin{aligned}
            \frac{1}{n} \sum_{i = 1}^n\big(\Theta_i -  B_i^{-1}\big)^+ & + \frac{1}{n} \sum_{i = 1}^{n} \big(1 - (\Theta_i - B_i^{-1})^+\big) \cdot H_i \cdot G_{n\Theta_i + 1} \\
            & + \frac{1}{n} \sum_{i = 1}^n \big(1 - (B_i - \Theta_i^{-1})^+ \big) \cdot H_i \cdot G_{nB_i + 1}
            ~,
        \end{aligned}
    \end{equation}
    where recall that $\Theta_i = \theta(\frac{i - 1}{n})$, $\Theta_i^{-1} = \theta^{-1}(\frac{i - 1}{n})$, $B_i = \beta(\frac{i - 1}{n})$, $B_i^{-1} = \beta^{-1}(\frac{i -1 }{n})$, $H_i = h(\frac{i -1 }{n})$, and $G_i = g(\frac{i -1 }{n})$ by our convention.
\end{lemma}

\begin{proof}
    We verify that the three terms in the lemma correspond to those in \Cref{lemma:universal-lower-bound}.
    
    Since $\theta$ and $\beta^{-1}$ are $n$-segment functions, we have
    \[
        \int_0^1 \Big( \theta(y_v) - \beta^{-1}(y_v) \Big)^+\d y_v 
        ~=~
        \sum_{i = 1}^{n} \int_\frac{i-1}{n}^\frac{i}{n} \Big( \theta(y_v) - \beta^{-1}(y_v) \Big)^+\d y_v 
        ~=~ 
        \frac{1}{n} \sum_{i = 1}^{n}\Big( \Theta_i - B^{-1}_i \Big)^+
        ~.
    \]
    
    Since $\theta$, $\beta^{-1}$, $h$ and $g$ are $n$-segment step functions, we have
    \begin{align*}
        & \int_0^1 \Big(1 - \big(\theta(y_v) - \beta^{-1}(y_v)\big)^+ \Big) \cdot h(y_v) \cdot  g(\theta(y_v)) \,\d y_v \\ 
        & \qquad
        =~ \sum_{i = 1}^{n} \int_\frac{i- 1}{n}^\frac{i}{n}  \Big(1 - \big(\theta(y_v) - \beta^{-1}(y_v)\big)^+ \Big) \cdot h(y_v) \cdot g(\theta(y_v)) \,\d y_v \\
        & \qquad
        =~ \frac{1}{n} \sum_{i = 1}^{n} \Big(1 - \big(\Theta_i - B^{-1}_i\big)^+ \Big) \cdot H_i \cdot G_{n\Theta_i + 1} 
        ~.
    \end{align*}

    Symmetrically, we can show that 
    \begin{equation*}
        \int_0^1 \Big(1 - \big(\beta(y_u) - \theta^{-1}(y_u)\big)^+\Big)\cdot h(y_u) \cdot g(\beta(y_u)) \,\d y_u 
        \,=\,
        \frac{1}{n} \sum_{i = 1}^n \big(1 - (B_i - \Theta_i^{-1})^+ \big) \cdot H_i \cdot G_{nB_i + 1}~. \qedhere
    \end{equation*}
\end{proof}

\subsection{Optimization-based Approach}
\label{Subsection:optimization-based-approach}

By the framework of randomized primal-dual analysis (\Cref{lemma:primal_dual}) and the bound in \Cref{lemma:discretization-lower-bound}, it remains to choose $n$-segment step functions such that \Cref{eqn:step-function-bound} is at least $F$ for any threshold ranks $\theta, \beta \in \mathcal{S}_n$, for the largest possible competitive ratio $F$.
This can be formulated as a QCQP below with $2n+1$ variables $F$, $G_1$ to $G_n$, and $H_1$ to $H_n$:

\begin{align*}
    \max \quad & F \\
    \text{s.t.} \quad & 
    \eqref{eqn:step-function-bound} ~\ge~ F 
    & & \forall \theta, \beta \in \mathcal{S}_n \\
    & H_i \cdot G_j + H_j \cdot G_i \le 1 & & \forall i, j \in [n] \\
    & G_1 \ge G_2 \ge \dots \ge G_n \ge G_{n+1} = 0 \\
    & 0 \le H_1 \le H_2 \le \dots \le H_n
\end{align*}

Our numerical results show that, up to appropriate scaling, the optimal solution for the above QCQP fits the general form given in \Cref{lemma:general-h-satisfy-constraint}.
To automate appropriate scaling, we solve the QCQP with two extra normalizing constraints:
\[
    G_1^2 + H_1^2 = 1 
    ~,\quad
    G_1 \ge H_1
    ~,
\]
so that $(G_1, H_1)$ lies on the unit circle, and below the 45-degree line $G = H$.

The problem has $2n +1$ variables and at most $O(4^{2n})$ constraints by \Cref{fact:num-of-seg-step-function}.
However, we observe that many of the constraints are redundant and can be removed to reduce the program size.
See \Cref{app:redundancy-reduction} for details.
Our first approach is to solve the QCQP using Gurobi~\cite{Gurobi} on an M2 Mac Pro with 24 physical cores and 192GB of memory.
At $n = 5$, the resulting competitive ratio is $0.6389$, which is already greater than $1-\frac{1}{e}$ and the bound obtained through the analytical method in \Cref{theorem:analytical-lower-bound}.
With this approach, $n = 7$ is the largest number of segments for which we can solve the QCQP directly, with a competitive ratio of $0.6487$.

For $n \ge 8$, both memory size and time become bottlenecks. 
Our second approach relies on the constraint-generation method, which is standard for solving programs with a smaller number of variables but a large number of constraints.
Concretely, we start with a relaxed program that only keeps a subset of the first constraint for $\theta(y) \equiv 1$.
After solving the relaxed program, we check the first constraints for all $\theta, \beta \in \mathcal{S}_n$, and construct the next relaxed program by adding back the constraints violated by at least half of the maximum violation.
We repeat this process until all constraints are satisfied. 

The constraint-generation approach allows us to solve QCQP relaxations with significantly fewer constraints.
For example, at $n = 7$, the original QCQP has about $2,000,000$ quadratic constraints, even after eliminating redundant combinations, including halving the number of combinations by the symmetry of $\theta$ and $\beta$.
In contrast, the constraint-generation approach only solves QCQPs with at most $57,000$ quadratic constraints, and finishes within one minute.
Using this approach, we are able to solve the QCQP for $n = 8$ and $9$.

\begin{table}[th]
    \renewcommand{\arraystretch}{1.15}
    \centering
    \caption{Summary of results from optimization-based methods}
    \label{tab:optimization-results}
    \medskip
    \begin{tabular}{lcccccc}
    \toprule
    Number of segments $n$ & 4 & 5 & 6 & 7 & 8 & 9\\
    \midrule
    Competitive ratio $F$ & $0.6321$ & $0.6389$ & $0.6447$ & $0.6487$ & $0.6515$ & $0.6537$ \\
    Time (direct) & $<$ 1 sec & 2 sec & 30 sec & 11 min \\
    \phantom{Time} (constr.\ gen.) & & & & 1 min & 32 min & 1 day \\
    \bottomrule
    \end{tabular}
\end{table}

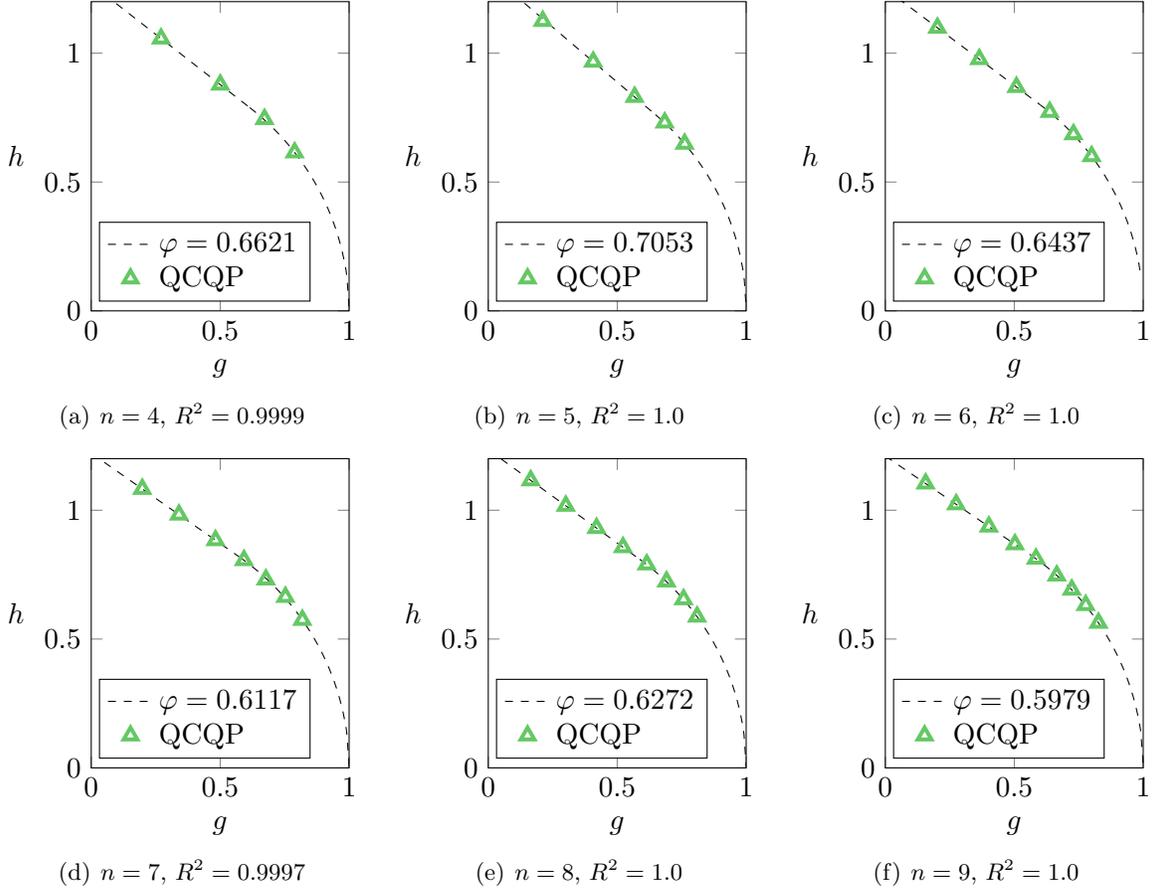
\begin{figure}[th]
    \centering
    \subfigure[$n=4$, $R^2 = 0.9999$]{
    \begin{tikzpicture}
        \begin{axis} [
            width=.4\textwidth,
            xmin=0,
            xmax=1.0,
            xtick distance=0.5,
            xlabel={$g$},
            ymin=0,
            ymax=1.2,
            ylabel={$h$},
            ylabel style={rotate=-90},
            unit vector ratio=1 1 1,
            legend pos=south west,
            legend cell align={left},
        ]
            \addplot+[forget plot,dashed,draw=black,domain=0.6147:1,smooth,samples=50,mark=none] {sqrt(1-x^2)};
            \addplot[dashed,draw=black,domain=0:0.6148] {(1 - x * 0.6147) / 0.7887};
            \addlegendentry{$\varphi = 0.6621$};
            \addplot+ [
                only marks,
                mark options={mark=triangle,color=myGreen,mark size=3,line width=1.5pt},
            ] coordinates {
                (0.7887, 0.6148)
                (0.6719, 0.7442)
                (0.5004, 0.8779)
                (0.2708, 1.0568)
            };
            \addlegendentry{QCQP};
        \end{axis}
    \end{tikzpicture}
    }
    \subfigure[$n=5$, $R^2 = 1.0$]{
    \begin{tikzpicture}
        \begin{axis} [
            width=.4\textwidth,
            xmin=0,
            xmax=1.0,
            xtick distance=0.5,
            xlabel={$g$},
            ymin=0,
            ymax=1.2,
            ylabel={$h$},
            ylabel style={rotate=-90},
            unit vector ratio=1 1 1,
            legend pos=south west,
            legend cell align={left},
        ]
            \addplot+[forget plot,dashed,draw=black,domain=0.6483:1,smooth,samples=50,mark=none] {sqrt(1-x^2)};
            \addplot[dashed,draw=black,domain=0:0.6483] {(1 - x * 0.6483) / 0.7614};
            \addlegendentry{$\varphi = 0.7053$};
            \addplot+ [
                only marks,
                mark options={mark=triangle,color=myGreen,mark size=3,line width=1.5pt},
            ] coordinates {
                (0.7614, 0.6483)
                (0.6842, 0.7308)
                (0.5675, 0.8302)
                (0.4069, 0.9669)
                (0.2113, 1.1250)                
            };
            \addlegendentry{QCQP};
        \end{axis}
    \end{tikzpicture}
    }
    \subfigure[$n=6$, $R^2 = 1.0$]{
    \begin{tikzpicture}
        \begin{axis} [
            width=.4\textwidth,
            xmin=0,
            xmax=1.0,
            xtick distance=0.5,
            xlabel={$g$},
            ymin=0,
            ymax=1.2,
            ylabel={$h$},
            ylabel style={rotate=-90},
            unit vector ratio=1 1 1,
            legend pos=south west,
            legend cell align={left},
        ]
            \addplot+[forget plot,dashed,draw=black,domain=0.6002:1,smooth,samples=50,mark=none] {sqrt(1-x^2)};
            \addplot[dashed,draw=black,domain=0:0.6002] {(1 - x * 0.6002) / 0.7999};
            \addlegendentry{$\varphi = 0.6437$};
            \addplot+ [
                only marks,
                mark options={mark=triangle,color=myGreen,mark size=3,line width=1.5pt},
            ] coordinates {
                (0.7998, 0.6002)
                (0.7296, 0.6853)
                (0.6373, 0.7720)
                (0.5085, 0.8687)
                (0.3644, 0.9768)  
                (0.2021, 1.0986)
            };
            \addlegendentry{QCQP};
        \end{axis}
    \end{tikzpicture}
    }

    \subfigure[$n=7$, $R^2 = 0.9997$]{
    \begin{tikzpicture}
        \begin{axis} [
            width=.4\textwidth,
            xmin=0,
            xmax=1.0,
            xtick distance=0.5,
            xlabel={$g$},
            ymin=0,
            ymax=1.2,
            ylabel={$h$},
            ylabel style={rotate=-90},
            unit vector ratio=1 1 1,
            legend pos=south west,
            legend cell align={left},
        ]
            \addplot+[forget plot,dashed,draw=black,domain=0.5743:1,smooth,samples=50,mark=none] {sqrt(1-x^2)};
            \addplot[dashed,draw=black,domain=0:0.5743] {(1 - x * 0.5743) / 0.8187};
            \addlegendentry{$\varphi = 0.6117$};
            \addplot+ [
                only marks,
                mark options={mark=triangle,color=myGreen,mark size=3,line width=1.5pt},
            ] coordinates {
                (0.8187, 0.5743)
                (0.7531, 0.6635)
                (0.6773, 0.7311)
                (0.5923, 0.8061)
                (0.4819, 0.8835)  
                (0.3403, 0.9828)
                (0.1980, 1.0826)
            };
            \addlegendentry{QCQP};
        \end{axis}
    \end{tikzpicture}
    }
    \subfigure[$n=8$, $R^2 = 1.0$]{
    \begin{tikzpicture}
        \begin{axis} [
            width=.4\textwidth,
            xmin=0,
            xmax=1.0,
            xtick distance=0.5,
            xlabel={$g$},
            ymin=0,
            ymax=1.2,
            ylabel={$h$},
            ylabel style={rotate=-90},
            unit vector ratio=1 1 1,
            legend pos=south west,
            legend cell align={left},
        ]
            \addplot+[forget plot,dashed,draw=black,domain=0.5869:1,smooth,samples=50,mark=none] {sqrt(1-x^2)};
            \addplot[dashed,draw=black,domain=0:0.5869] {(1 - x * 0.5869) / 0.8097};
            \addlegendentry{$\varphi = 0.6272$};
            \addplot+ [
                only marks,
                mark options={mark=triangle,color=myGreen,mark size=3,line width=1.5pt},
            ] coordinates {
                (0.8096, 0.5869)
                (0.7573, 0.6542)
                (0.6912, 0.7234)
                (0.6145, 0.7896)
                (0.5224, 0.8564)  
                (0.4198, 0.9308)
                (0.3007, 1.0171)
                (0.1646, 1.1158)
            };
            \addlegendentry{QCQP};
        \end{axis}
    \end{tikzpicture}
    }
    \subfigure[$n=9$, $R^2 = 1.0$]{
    \begin{tikzpicture}
        \begin{axis} [
            width=.4\textwidth,
            xmin=0,
            xmax=1.0,
            xtick distance=0.5,
            xlabel={$g$},
            ymin=0,
            ymax=1.2,
            ylabel={$h$},
            ylabel style={rotate=-90},
            unit vector ratio=1 1 1,
            legend pos=south west,
            legend cell align={left},
        ]
            \addplot+[forget plot,dashed,draw=black,domain=0.5629:1,smooth,samples=50,mark=none] {sqrt(1-x^2)};
            \addplot[dashed,draw=black,domain=0:0.5629] {(1 - x * 0.5629) / 0.8265};
            \addlegendentry{$\varphi = 0.5979$};
            \addplot+ [
                only marks,
                mark options={mark=triangle,color=myGreen,mark size=3,line width=1.5pt},
            ] coordinates {
                (0.8265, 0.5629)
                (0.7769, 0.6318)
                (0.7230, 0.6916)
                (0.6648, 0.7465)
                (0.5842, 0.8120)  
                (0.5025, 0.8677)
                (0.4016, 0.9364)
                (0.2748, 1.0228)
                (0.1563, 1.1034)
            };
            \addlegendentry{QCQP};
        \end{axis}
    \end{tikzpicture}
    }    
    \caption{The optimal $n$-segment step functions $g$ and $h$ with the rightmost mark begin $(G_1, H_1)$, the general form in \Cref{lemma:general-h-satisfy-constraint} with the best $\varphi$, and the coefficient of determination ($R^2$)}
    \label{fig:optimization-results}
\end{figure}

\Cref{tab:optimization-results} summarizes our optimization-based results.
\Cref{fig:optimization-results} shows the optimal $n$-segment step functions $g$ and $h$ for $4 \le n \le 9$, and the best value of $\varphi$ for fitting their values according to the general form presented in \Cref{lemma:general-h-satisfy-constraint}.

\subsection{Verification-based Approach: Proof of Main Theorem}
\label{Subsection:verification-based-approach}

Finally, given the highly structured optimal step functions $g$ and $h$ presented in \Cref{fig:optimization-results}, we can guess approximately optimal choices of $n$-segment step functions $g$ and $h$ for larger values of $n$.
Iterating over all combinations of marginal ranks $\theta$ and $\beta$ is much faster than solving the optimization problem. 
Formally, for any given $n$-segment step functions $g$ and $h$, by \Cref{lemma:discretization-lower-bound}, we compute the resulting competitive ratio as
\begin{align*}
    F = \min_{\theta, \beta\in \mathcal S_n} \eqref{eqn:step-function-bound}~.
\end{align*}

We can verify the competitive ratio of our guessed $g$ and $h$ within a day up to $n = 13$.
Computing the ratio from $n = 14$ takes more than 10 days and only leads to an improvement of less than $10^{-4}$.
\Cref{tab:main-theorem} presents the guessed $13$-segment step functions $g$ and $h$, leading to a verified competitive ratio of $0.659$ and proving \Cref{thm:main}.
\Cref{tab:verification-results} demonstrates the competitive ratios from the verification-based approach for $7 \le n \le 13$.

\begin{table}[th]
    \renewcommand{\arraystretch}{1.15}
    \centering
    \caption{Summary of results from verification-based method}
    \label{tab:verification-results}
    \medskip
    \begin{tabular}{lccccccc}
    \toprule
    Number of segments $n$ & 7 & 8 & 9 & 10 & 11 & 12 & 13 \\
    \midrule
    Competitive ratio $F$ & $0.6479$ & $0.6506$ & $0.6529$ & $0.6549$ & $0.6565$ & $0.6575$ & $0.6590$ \\
    Time & \multicolumn{2}{c}{$<$ 1 sec} & 2 sec & 26 sec & 6 min & 90 min & 1 day \\
    \bottomrule
    \end{tabular}
\end{table}

\begin{table}[th]
    \renewcommand{\arraystretch}{1.15}
    \centering
    \caption{Guessed $13$-segment step functions $g$ and $h$ (up to $10^{-4}$) for proving \Cref{thm:main}}
    \label{tab:main-theorem}
    \medskip
    \begin{tabular}{@{\hskip 0.1in}c@{\hskip 0.5in}c@{\hskip 0.5in}c@{\hskip 0.1in}}
    \toprule
    $i$ & $G_i$ & $H_i$ \\
    \midrule
    $1$ & $0.8200$ & $0.5724$ \\
    $2$ & $0.7883$ & $0.6152$ \\
    $3$ & $0.7530$ & $0.6580$ \\
    $4$ & $0.7139$ & $0.7002$ \\
    $5$ & $0.6708$ & $0.7416$ \\
    $6$ & $0.6237$ & $0.7817$ \\
    $7$ & $0.5724$ & $0.8200$ \\
    $8$ & $0.5152$ & $0.8599$ \\
    $9$ & $0.4498$ & $0.9055$ \\
    $10$ & $0.3763$ & $0.9569$ \\
    $11$ & $0.2945$ & $1.0140$ \\
    $12$ & $0.2045$ & $1.0767$ \\
    $13$ & $0.1064$ & $1.1453$ \\
    \bottomrule
    \end{tabular}
\end{table}

\section{Hardness Results}
\label{sec:hardness_bipartite}

In this section, we provide upper bounds for any (randomized) algorithms for the oblivious bipartite matching problem, showing that no (randomized) algorithm can achieve a competitive ratio strictly larger than $0.7961$.
All hard instances we provide in this section are unweighted.
Thus, the upper bounds apply to all three models: unweighted, vertex-weighted, and edge-weighted.
Moreover, our bounds hold under the assumption that the underlying graph (but not the identities of the vertices) is known to the algorithm.
We also extend our hard instance and analysis to general graphs and provide an upper bound of $0.7583$.

\subsection{Warm-up Analysis}
\label{ssec:warm-up-bipartite}

We first do a warm-up analysis by giving a hard instance with four vertices, for which no algorithm can achieve a competitive ratio strictly larger than $7/8 = 0.875$.
The result will also serve as a building block for the later analysis (for establishing the upper bound of $0.8241$).

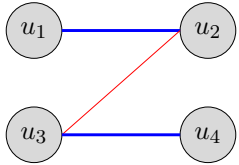
\begin{figure}[htb]
\begin{center}
    \resizebox{0.2\textwidth}{!}{
    \begin{tikzpicture}

	\draw [fill = gray!30] (2.5,1.5) circle (0.4); 
        \node at (2.5,1.5) {$u_1$};
        \draw [fill = gray!30] (2.5,0) circle (0.4); 
        \node at (2.5,0) {$u_3$};
        
	\draw [fill = gray!30] (5,1.5) circle (0.4); 
        \node at (5,1.5) {$u_2$};
	\draw [fill = gray!30] (5,0) circle (0.4); 
        \node at (5,0) {$u_4$};
        
	\draw [color=blue, very thick] (2.9,1.5)--(4.6,1.5); %
	\draw [color=red] (2.9,0)--(4.6,1.5); %
	\draw [color=blue, very thick] (2.9,0)--(4.6,0); %
    \end{tikzpicture} }
\end{center}
\caption{A warm-up hard instance for bipartite graphs, where the edges from the perfect matching are presented in thick blue color while the other edge is colored red.}
\label{fig:warmup-hardness-bipartite}
\end{figure}

Recall that the algorithm can see two vertices on the left and two on the right, but not the edges between them.
By randomly shuffling the vertices and using Yao's Lemma, it suffices to analyze the expected performance of any deterministic algorithm, over the random instances.
Note that the deterministic algorithm can be \emph{adaptive}, e.g, depending on the result of the past queries, the next edge to be queried can be different.

\begin{lemma} \label{lemma:warm-up-bipartite}
    No algorithm can achieve a competitive ratio strictly larger than $7/8 = 0.875$ for the (unweighted) oblivious matching problem with the instance shown in Figure~\ref{fig:warmup-hardness-bipartite}. 
\end{lemma}
\begin{proof}
    We use $\{a,c\}$ to name the two vertices on the left, $\{b,d\}$ for the two on the right (whose identities are unknown), and assume w.l.o.g. that the algorithm queries $(a,b)$ first.
    Clearly, the first edge $(a,b)$ being queried is uniformly distributed in $\{ (u_1,u_2), (u_1,u_4), (u_3,u_2), (u_3,u_4) \}$ due to the random shuffling, among which only $(u_1,u_4)$ is a \emph{null edge} (e.g., the query is unsuccessful).
    \begin{itemize}
        \item If the first query is successful and the edge is either $(u_1,u_2)$ or $(u_3,u_4)$ (which happens w.p. $1/2$), then the algorithm will match two edges;
        \item If the first query is successful and the edge is $(u_3,u_2)$ (which happens w.p. $1/4$), then the algorithm can match only one edge;
        \item If the first query is unsuccessful (which happens w.p. $1/4$), we are certain that $a = u_1$ and $b = u_4$. Therefore by querying $(a,d)$ and $(c,b)$ next, the algorithm matches two edges.
    \end{itemize}
    
    In summary, the expected number of edges the algorithm matches is $\frac{3}{4}\times 2 + \frac{1}{4}\times 1 = \frac{7}{4}$, which implies a competitive ratio of $\frac{7}{8}$ since the maximum matching has size $2$.
\end{proof}

It can be shown that the competitive ratio of the Ranking algorithm on the above instance is $0.875$, and thus is optimal.

\subsection{A Better Upper Bound}

Next we consider the instance shown in Figure~\ref{fig:hardness-bipartite}, which has six vertices and admits a unique perfect matching $\{ (u_1,u_2),(u_3,u_4),(u_5,u_6) \}$.
We show that no algorithm can do better than $89/108$-approximate for this instance.
Interestingly, it can be shown that the competitive ratio of Ranking on the above instance is $89/108$, and is optimal again.

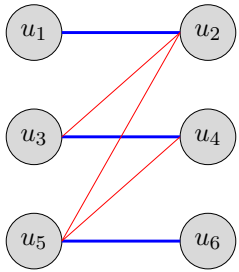
\begin{figure}[htb]
\begin{center}
    \resizebox{0.2\textwidth}{!}{
    \begin{tikzpicture}

	\draw [fill = gray!30] (2.5,3) circle (0.4); 
        \node at (2.5,3) {$u_1$};
	\draw [fill = gray!30] (2.5,1.5) circle (0.4); 
        \node at (2.5,1.5) {$u_3$};
        \draw [fill = gray!30] (2.5,0) circle (0.4); 
        \node at (2.5,0) {$u_5$};
        
	\draw [fill = gray!30] (5,3) circle (0.4); 
        \node at (5,3) {$u_2$};
	\draw [fill = gray!30] (5,1.5) circle (0.4); 
        \node at (5,1.5) {$u_4$};
	\draw [fill = gray!30] (5,0) circle (0.4); 
        \node at (5,0) {$u_6$};
        
	\draw [color=blue, very thick] (2.9,3)--(4.6,3); %
	\draw [color=blue, very thick] (2.9,1.5)--(4.6,1.5); %
	\draw [color=blue, very thick] (2.9,0)--(4.6,0); %
	\draw [color=red] (2.9,1.5)--(4.6,3); %
	\draw [color=red] (2.9,0)--(4.6,1.5); %
	\draw [color=red] (2.9,0)--(4.6,3); %
    \end{tikzpicture} }
\end{center}
\caption{Hard instance for bipartite graphs: There are six edges in the instance, among which $(u_1,u_2),(u_3,u_4),(u_5,u_6)$ appear in the (unique) perfect matching, and we call them \emph{good} edges (represented by the thick blue edges). We call the other three edges $(u_3,u_2),(u_5,u_2),(u_5,u_4)$ \emph{bad} edges (represented by the red edges).
For completeness, we call $(u_1,u_4),(u_1,u_6),(u_3,u_6)$ \emph{null} edges.}
\label{fig:hardness-bipartite}
\end{figure}

\begin{theorem} \label{theorem:hardness-bipartite}
    No algorithm can achieve a competitive ratio strictly larger than $\frac{89}{108}\leq 0.8241$ for the (unweighted) oblivious matching problem with the instance shown in Figure~\ref{fig:hardness-bipartite}.
\end{theorem}
\begin{proof}
    We use $\{a,c,e\}$ and $\{b,d,f\}$ to name the three vertices on the left and right, respectively, and assume w.l.o.g. that $(a,b)$ is the first queried edge.
    \begin{itemize}
        \item If $(a,b)$ is a bad edge (which happens w.p.\ $1/3$), then algorithm (with any follow-up queries) is guaranteed to match exactly two edges;
        \item If $(a,b)$ is a good edge (which happens w.p.\ $1/3$), then the remaining instance (with vertices $a,b$ removed) is the same as the one shown in Figure~\ref{fig:warmup-hardness-bipartite}, for which no algorithm can match more than $7/4$ edges in expectation;
        \item The case when $(a,b)$ is a null edge is slightly more complex: while the algorithm does not match any edge nor conclude the identity of $a$ or $b$, there is useful information the algorithm can elicit to decide an optimal follow-up query sequence. We show in Lemma~\ref{lemma:first-edge-is-null-bipartite} that the expected number of edges the optimal query sequence can match is at most $8/3$.
    \end{itemize}

    \begin{lemma} \label{lemma:first-edge-is-null-bipartite}
        If the first edge $(a,b)$ queried by the algorithm is a null edge, the expected number of edges matched by the algorithm is at most $8/3$. 
    \end{lemma}

    We defer the proof of the lemma to the end of this section.
    Given the lemma, we can derive an upper bound on the competitive ratio of the algorithm as follows:
    \begin{equation*}
        \frac{1}{3}\times \left( \frac{1}{3}\times 2 + \frac{1}{3}\times (1+\frac{7}{4}) + \frac{1}{3}\times \frac{8}{3} \right) = \frac{89}{108} \leq 0.8241. \qedhere
    \end{equation*}
\end{proof}

Finally, we prove Lemma~\ref{lemma:first-edge-is-null-bipartite}.

\begin{proofof}{Lemma~\ref{lemma:first-edge-is-null-bipartite}}
    Recall that at the moment the algorithm has identified a null edge $(a,b)$, which could be any of $\{ (u_1,u_4),(u_1,u_6),(u_3,u_6) \}$, each with probability $1/3$.

    The remaining $8$ pairs of vertices can be categorized into two types:
    \begin{itemize}
        \item The pairs $\{(c,d),(c,f),(e,d),(e,f)\}$ have no endpoint in $\{a,b\}$. We call them \emph{fresh} pairs;
        \item Each of $\{(a,d),(a,f),(c,b),(e,b)\}$ has one endpoint in $\{a,b\}$. We call them \emph{extending} pairs.
    \end{itemize}

    \begin{figure}[htb]
    \begin{center}
        \resizebox{0.2\textwidth}{!}{
        \begin{tikzpicture}
    	\draw [fill = gray!30] (2.5,3) circle (0.4); 
            \node at (2.5,3) {$a$};
    	\draw [fill = gray!30] (2.5,1.5) circle (0.4); 
            \node at (2.5,1.5) {$c$};
            \draw [fill = gray!30] (2.5,0) circle (0.4); 
            \node at (2.5,0) {$e$};
    	\draw [fill = gray!30] (5,3) circle (0.4); 
            \node at (5,3) {$b$};
    	\draw [fill = gray!30] (5,1.5) circle (0.4); 
            \node at (5,1.5) {$d$};
    	\draw [fill = gray!30] (5,0) circle (0.4); 
            \node at (5,0) {$f$};
    	\draw [dotted] (2.9,3)--(4.6,3); %
    	\draw (2.9,3)--(4.6,1.5); %
    	\draw (2.9,3)--(4.6,0); %
    	\draw (2.9,1.5)--(4.6,3); %
    	\draw (2.9,0)--(4.6,3); %
    	\draw [dashed] (2.9,1.5)--(4.6,1.5); %
    	\draw [dashed] (2.9,1.5)--(4.6,0); %
    	\draw [dashed] (2.9,0)--(4.6,1.5); %
    	\draw [dashed] (2.9,0)--(4.6,0); %
        \end{tikzpicture} }
    \end{center}
    \caption{The current configuration: $(a,b)$ has been identified as a null edge. The dashed edges are fresh pairs while the solid edges are extending pairs.}
    \label{fig:configuration-after-null-bipartite}
    \end{figure}
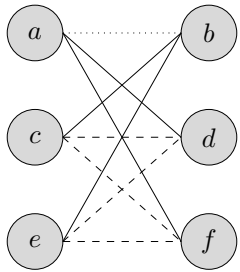

    See Figure~\ref{fig:configuration-after-null-bipartite} for an illustration of the current configuration.
    Due to the symmetry of the instance and the random shuffling, we only need to consider two types of algorithms: one that queries a fresh pair next and the other queries an extending pair next.

    We first show that the algorithm that queries a fresh pair matches at most $5/2$ edges in expectation, which is strictly smaller than $8/3$.
    Assume w.l.o.g. that the pair being queried is $(c,d)$.
    We show that the probability of $(c,d)$ being a bad edge is at least $1/2$.
    Recall that $(a,b)$ could be any of the three null edges $\{ (u_1,u_4),(u_1,u_6),(u_3,u_6) \}$.
    It can be verified that no matter which null edge $(a,b)$ is, among the four fresh pairs, at least two (sometimes three) of them are bad. Since the vertices are randomly shuffled, we have $\Pr[(c,d) \text{ is bad}] \geq 1/2$.
    Therefore, the expected number of edges matched by the algorithm (with any query sequence following $(a,b)$ and $(c,d)$) is at most
    \begin{equation*}
        \Pr[(c,d) \text{ is bad}]\times 2 + \Pr[(c,d) \text{ is not bad}]\times 3 \leq \frac{5}{2} < \frac{8}{3}.
    \end{equation*}

    Next we consider the algorithm that queries an extending pair, say $(a,d)$ (w.l.o.g.).
    \begin{itemize}
        \item The probability of $(a,d)$ being a bad edge is
        \begin{equation*}
            \frac{1}{3}\times 0 + \frac{1}{3}\times  0 + \frac{1}{3}\times \frac{1}{2} = \frac{1}{6},
        \end{equation*}
        where the three terms on the LHS correspond to the probability of $(c,d)$ being bad for the three cases of $(a,b)\in \{ (u_1,u_4),(u_1,u_6),(u_3,u_6) \}$.
        Note that if $(a,d)$ is bad then the algorithm is doomed to match exactly two edges;

        \item Similarly, the probability of $(a,d)$ being a null edge is
        \begin{equation*}
            \frac{1}{3}\times \frac{1}{2} + \frac{1}{3}\times \frac{1}{2} + \frac{1}{3}\times 0 = \frac{1}{3},
        \end{equation*}
        in which case the algorithm can assert that $a = u_1$ and $f = u_2$ (and match them with query $(a,f)$).
        However, the two vertices $\{c,e\}$ (resp. $\{b,d\}$) are still indistinguishable.
        Therefore, by the analysis we presented in Section~\ref{ssec:warm-up-bipartite}, the expected number of edges matched by the algorithm in this case is $1+7/4 = 11/4$.

        \item Finally, the probability of $(a,d)$ being a good edge is
        \begin{equation*}
            \frac{1}{3}\times \frac{1}{2} + \frac{1}{3}\times \frac{1}{2} + \frac{1}{3}\times \frac{1}{2} = \frac{1}{2},
        \end{equation*}
        in which case the algorithm matches $(a,d)$ and has four vertices $\{c,e\}\cup \{b,f\}$ remains.
        However, in this case $b$ and $f$ are different\footnote{For instance, $b$ cannot be $u_2$, and is more likely to be $u_6$ than being $u_4$.} and we distinguish the two algorithms, one that queries a pair involving $b$ while the other queries a pair involving $f$.
        Recall that so far the algorithm has identified a null edge $(a,b)$ and matched an edge $(a,d)$.
        While the algorithm is not informed of whether $(a,d)$ is good or bad, it is in fact w.l.o.g. to assume that it is a good edge, because for the case when $(a,d)$ is bad, such assumption is harmless (the algorithm is doomed to match two edges in total).
        Under this assumption, we can compute the probability of an edge incident to $b$ (resp. $f$) being bad:
        \begin{align*}
            \Pr[(c,d) \text{ is bad}] & = \frac{1}{3}\times \frac{1}{2} + \frac{1}{3}\times 0 + \frac{1}{3}\times 0 = \frac{1}{6}, \\
            \Pr[(e,f) \text{ is bad}] & = \frac{1}{3}\times 0 + \frac{1}{3}\times \frac{1}{2} + \frac{1}{3}\times \frac{1}{2} = \frac{1}{3}.
        \end{align*}
        Consequently, it is optimal to query an edge incident to $b$ next, which results in an expected total number of matched edges of $\frac{1}{6}\times 2 + \frac{5}{6} \times 3 = \frac{17}{6}$.
    \end{itemize}

    In summary, the expected number of matched edges is at most $\frac{1}{6}\times 2 + \frac{1}{3}\times \frac{11}{4} + \frac{1}{2}\times \frac{17}{6} = \frac{8}{3}$. 
\end{proofof}

\subsection{Computer-Aided Upper-Bound Analysis}

By our previous analysis, it is reasonable to expect an even smaller upper bound by further generalizing the hard instances in Figure~\ref{fig:hardness-bipartite} to graphs with more vertices.
Specifically, for any $n \ge 4$, we can construct hard instances with $n$ vertices on the left (namely $u_1, u_3, \dots, u_{2n - 1}$) and another $n$ vertices on the right (namely $u_2, u_4, \dots, u_{2n}$). 
Edges exist between $u_i\in L$ and $u_j\in R$ if and only if $i \geq j-1$. 
We call the hard instance $\mathcal H_n$.

However, we found that even $\mathcal H_4$ is already quite challenging to analyze.
Therefore, we adopt a backward dynamic programming approach that generalizes our previous analysis framework.
At any point in time during the execution of an algorithm, the outcome of the past queries defines the current configuration, based on which we can calculate the expected outcome of every unqueried edge by enumerating all possible embeddings of the underlying graph that agrees with the current configuration.
We introduce the details as follows.

Given the result of past queries, we classify every pair $(u,v)$ into three types, namely, (1) existing edge if $(u,v)$ is queried and the edge exists; (2) null edge if $(u,v)$ is queried but the edge does not exist; and (3) unqueried edge.
We call the current configuration a \emph{state} of the algorithm.
Our goal is to compute the expected number of edges matched by the \emph{optimal} strategy, under the current state.
Since the structure of the underlying graph is known but the identities of the vertices are not, the algorithm can first enumerate all $(2n)!$ possible embeddings of graph, and then keep only those that agrees with the current state.
Note that each of the kept embeddings happen with the same probability.
Then for every unqueried edge between two unmatched vertices, we can calculate the probability for its existence if it is hypothetically queried, which corresponds to the \emph{transit probability} to another state.
Note that the final state is when all edges between the unmatched vertices are queried, in which case the total number of matched edges are finalized.
 
Therefore under every state, the optimal strategy is to choose the edge that gives the maximum expected number of matched edges.
See Figure~\ref{fig:state-transition} for an illustrating example when $n=2$.

\begin{figure}[htb]
\begin{center}
    \resizebox{\textwidth}{!}{
    \begin{tikzpicture}
        \tikzset{rect/.style={rounded corners=3pt}}
        
        \node at (3.75, -1) {Current State};
        \node at (-2.15, -1.8) {Possible embeddings};

        \draw[color = blue, very thick] (11, 1.3) -- (11.3, 1.3) node[black, right, scale = 0.7] {Existing};
        \draw[color = black, very thick, dashed] (11, 0.75) -- (11.3, 0.75) node[black, right, scale = 0.7] {Not existing} ;
        \draw[color = black] (11, 0.2) -- (11.3, 0.2) node[black, right, scale = 0.7] {Unquried} ;

	\draw [fill = gray!30] (2.5,1.5) circle (0.4); 
        \node at (2.5,1.5) {$a$};
        \draw [fill = gray!30] (2.5,0) circle (0.4); 
        \node at (2.5,0) {$c$};
        
	\draw [fill = gray!30] (5,1.5) circle (0.4); 
        \node at (5,1.5) {$b$};
	\draw [fill = gray!30] (5,0) circle (0.4); 
        \node at (5,0) {$d$};
        
	\draw [color=black] (2.9,1.5)--(4.6,1.5); %
        \draw [color=black] (2.9,1.5)--(4.6,0); %
	\draw [color=black] (2.9,0)--(4.6,0); %
	\draw [color=blue, very thick] (2.9,0)--(4.6,1.5); %

        \draw [fill = gray!30] (8,2.5) circle (0.3); 
        \node at (8, 2.5) {$a$};
        \draw [fill = gray!30] (8,1.5) circle (0.3); 
        \node at (8, 1.5) {$c$};
        
	\draw [fill = gray!30] (9.5,2.5) circle (0.3); 
        \node at (9.5, 2.5) {$b$};
	\draw [fill = gray!30] (9.5,1.5) circle (0.3); 
        \node at (9.5, 1.5) {$d$};
        
	\draw [color=black] (8.3, 2.5)--(9.2,2.5); %
        \draw [color=blue, very thick] (8.3, 2.5)--(9.2,1.5); %
	\draw [color=black] (8.3,1.5)--(9.2,1.5); %
	\draw [color=blue, very thick] (8.3, 1.5)--(9.2,2.5); %

        \draw[black, very thick, ->] (9.9, 2) -- (10.2, 2) node [right] {$\dots$};

        \draw [fill = gray!30] (8,0) circle (0.3); 
        \node at (8, 0) {$a$};
        \draw [fill = gray!30] (8,-1) circle (0.3); 
        \node at (8, -1) {$c$};
        
	\draw [fill = gray!30] (9.5,0) circle (0.3); 
        \node at (9.5, 0) {$b$};
	\draw [fill = gray!30] (9.5,-1) circle (0.3); 
        \node at (9.5, -1) {$d$};

	\draw [color=black] (8.3, 0)--(9.2,0); %
        \draw [color=black, very thick, dashed] (8.3, 0)--(9.2,-1); %
	\draw [color=black] (8.3,-1)--(9.2,-1); %
	\draw [color=blue, very thick] (8.3, -1)--(9.2,0); %

        \draw[black, very thick, ->] (9.9, -0.5) -- (10.2, -0.5) node [right] {$\dots$};

        \draw [fill = gray!30] (-1.5,2.5) circle (0.3); 
        \node at (-1.5, 2.5) {$a$};
        \draw [fill = gray!30] (-1.5,1.5) circle (0.3); 
        \node at (-1.5, 1.5) {$c$};
        
	\draw [fill = gray!30] (0,2.5) circle (0.3); 
        \node at (0, 2.5) {$b$};
	\draw [fill = gray!30] (0,1.5) circle (0.3); 
        \node at (0, 1.5) {$d$};

	\draw [color=black, very thick, dashed] (-1.2, 2.5)--(-0.3,2.5); %
        \draw [color=blue, very thick] (-1.2, 2.5)--(-0.3, 1.5); %
	\draw [color=blue, very thick] (-1.2,1.5)--(-0.3, 1.5); %
	\draw [color=blue, very thick] (-1.2,1.5)--(-0.3, 2.5); %

        \draw [fill = gray!30] (-4.3,2.5) circle (0.3); 
        \node at (-4.3, 2.5) {$a$};
        \draw [fill = gray!30] (-4.3,1.5) circle (0.3); 
        \node at (-4.3, 1.5) {$c$};
        
	\draw [fill = gray!30] (-2.8,2.5) circle (0.3); 
        \node at (-2.8, 2.5) {$b$};
	\draw [fill = gray!30] (-2.8,1.5) circle (0.3); 
        \node at (-2.8, 1.5) {$d$};

	\draw [color=blue, very thick] (-4, 2.5)--(-3.1,2.5); %
        \draw [color=blue, very thick] (-4, 2.5)--(-3.1,1.5); %
	\draw [color=black, very thick, dashed] (-4,1.5)--(-3.1,1.5); %
	\draw [color=blue, very thick] (-4,1.5)--(-3.1,2.5); %

        \draw [fill = gray!30] (-2.9,0) circle (0.3); 
        \node at (-2.9, 0) {$a$};
        \draw [fill = gray!30] (-2.9,-1) circle (0.3); 
        \node at (-2.9,-1) {$c$};
        
	\draw [fill = gray!30] (-1.4,0) circle (0.3); 
        \node at (-1.4,0) {$b$};
	\draw [fill = gray!30] (-1.4,-1) circle (0.3); 
        \node at (-1.4,-1) {$d$};

	\draw [color=blue, very thick] (-2.6, 0)--(-1.7,0); %
        \draw [color=black, very thick, dashed] (-2.6, 0)--(-1.7,-1); %
	\draw [color=blue, very thick] (-2.6,-1)--(-1.7,-1); %
	\draw [color=blue, very thick] (-2.6, -1)--(-1.7,0); %

        \draw[double, double distance=0.5pt, line width=1.2pt, ->] (1.6,0.75) -- (0.8,0.75);

        \draw[black, very thick] (5.8, 0.75) -- (7, 0.75);
        \node[scale = 0.7] at (6.4, 1.4) {Query};
        \node[scale = 0.7] at (6.4, 1) {$(a, d)$};

        \draw[black, very thick] (7, 0.75) arc (270:360:0.2);
        \draw[black, very thick] (7.2, 0.95) -- (7.2, 1.8);
        \draw[black, very thick] (7.2, 1.8) arc (180:90:0.2);
        \draw[black, very thick, ->] (7.4, 2) -- (7.6, 2); 
        \node [scale = 0.7] at (7, 2.4) {w.p. $\displaystyle\frac{2}{3}$};
        
        \draw[black, very thick] (7, 0.75) arc (90:0:0.2);
        \draw[black, very thick] (7.2, 0.55) -- (7.2, -0.3);
        \draw[black, very thick] (7.2, -0.3) arc (180:270:0.2);
        \draw[black, very thick, ->] (7.4, -0.5) -- (7.6, -0.5); 
        \node [scale = 0.7] at (7, -0.9) {w.p. $\displaystyle\frac{1}{3}$};
    \end{tikzpicture} }
\end{center}
\caption{State transition: Given the current state, the figure shows the transition when querying $(a, d)$. Conditioned on $(c, b)$ being an existing edge, we can list the possible embedding of the underlying graph that agrees with the current state on the right, and compute the expected outcome if we query $(a,d)$ next.}
\label{fig:state-transition}
\end{figure}
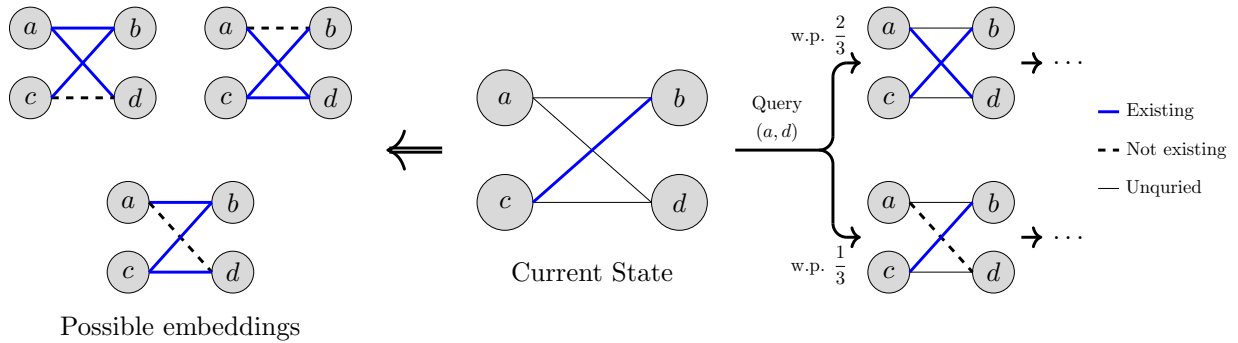

\begin{theorem}\label{theorem:hardness-dp}
    No algorithm can achieve a competitive ratio strictly larger than $0.7961$ for the (unweighted) oblivious bipartite matching problem with the instance $\mathcal H_6$. 
\end{theorem}
\begin{table}[th]
    \renewcommand{\arraystretch}{1.15}
    \centering
    \caption{Summary of upper bounds from $\mathcal H_3$ to $\mathcal H_6$.}
    \label{tab:upper-bounds-results}
    \medskip
    \begin{tabular}{lccccc}
    \toprule
    Number of vertices on each side $n$ & 3 & 4 & 5 & 6 \\
    \midrule
    Upper bound ratio  & $0.8241$ & $0.8047$ & $0.7981$ & $0.7961$ \\
    Time & \multicolumn{2}{c}{$<$ 1 sec}  & 1 sec & 3 min \\
    \bottomrule
    \end{tabular}
\end{table}

We run the DP program on an M2 Mac Pro with 24 physical cores and 192GB memory. \Cref{tab:upper-bounds-results} demonstrates the competitive ratios for $3 \le n \le 6$.  
It can be verified that the competitive ratios of Ranking on these instances match the upper bounds. Thus Ranking is optimal again.

\subsection{Upper Bound for General Graphs}
\label{ssec:hardness_general}

One advantage of our dynamic programming approach is that it does not rely on the graph being bipartite. 
Therefore, we can apply the same technique to upper bound the competitive ratios for general graphs.
In the following, we construct a non-bipartite hard instance and establish an upper bound of $0.7583$ for the competitive ratio of any (randomized adaptive) algorithm.

Consider the hard instance $\mathcal{\hat{H}}_n$ with $2n$ vertices $\{ u_1,u_2,\ldots,u_{2n} \}$, where there is an edge between $u_i$ and $u_j$ if and only if $i$ is odd and $i \geq j-1$.
See Figure~\ref{fig:hardness-general} for the examples of $\mathcal{\hat{H}}_2$ and $\mathcal{\hat{H}}_3$.

\begin{figure}[th]
    \centering
    \subfigure{
    \begin{tikzpicture}
        \draw [fill = gray!30] (2.5,1.5) circle (0.4); 
        \node at (2.5,1.5) {$u_1$};
        \draw [fill = gray!30] (2.5,0) circle (0.4); 
        \node at (2.5,0) {$u_3$};
        
	\draw [fill = gray!30] (5,1.5) circle (0.4); 
        \node at (5,1.5) {$u_2$};
	\draw [fill = gray!30] (5,0) circle (0.4); 
        \node at (5,0) {$u_4$};
        
	\draw [color=blue, very thick] (2.9,1.5)--(4.6,1.5); %
	\draw [color=blue, very thick] (2.9,0)--(4.6,0); %
	\draw [color=red] (2.5,1.1)--(2.5,0.4); %
	\draw [color=red] (2.83,0.2)--(4.6,1.5); %
    \end{tikzpicture}
    }
    \qquad\qquad
    \subfigure{
    \begin{tikzpicture}
	\draw [fill = gray!30] (2.5,3) circle (0.4); 
        \node at (2.5,3) {$u_1$};
	\draw [fill = gray!30] (2.5,1.5) circle (0.4); 
        \node at (2.5,1.5) {$u_3$};
        \draw [fill = gray!30] (2.5,0) circle (0.4); 
        \node at (2.5,0) {$u_5$};
	\draw [fill = gray!30] (5,3) circle (0.4); 
        \node at (5,3) {$u_2$};
	\draw [fill = gray!30] (5,1.5) circle (0.4); 
        \node at (5,1.5) {$u_4$};
	\draw [fill = gray!30] (5,0) circle (0.4); 
        \node at (5,0) {$u_6$};
        
	\draw [color=blue, very thick] (2.9,3)--(4.6,3); %
	\draw [color=blue, very thick] (2.9,1.5)--(4.6,1.5); %
	\draw [color=blue, very thick] (2.9,0)--(4.6,0); %
	\draw [color=red] (2.5,2.6)--(2.5,1.9); %
	\draw [color=red] (2.5,1.1)--(2.5,0.4); %
	\draw [color=red] (2.2,2.75) to [out=225,in=135] (2.2,0.25); %
	\draw [color=red] (2.83,1.75)--(4.6,3); %
	\draw [color=red] (2.83,0.2)--(4.6,1.5); %
	\draw [color=red] (2.75,0.32)--(4.6,3); %
    \end{tikzpicture}
    }
    \caption{Examples of hard instances $\mathcal{\hat{H}}_n$ for $n=2,3$.}
    \label{fig:hardness-general}
\end{figure}
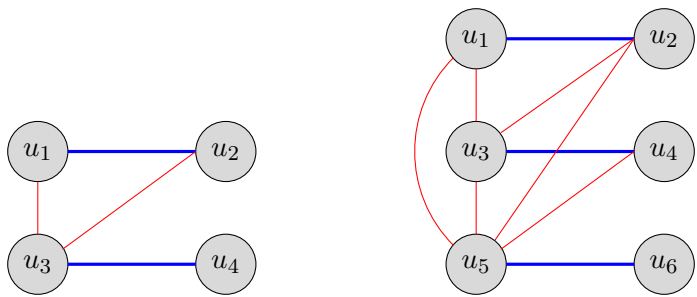

We remark that the hard instance $\mathcal{\hat{H}}_2$ was considered in~\cite{GT12}, where the state-of-the-art upper bound of $19/24 \approx 0.7917$ was established.
Deploying our dynamic programming approach, we show an improved upper bound of $91/120 \approx 0.7583$, with the hard instance being $\mathcal{\hat{H}}_3$.
Interestingly, it can be verified that Ranking achieves competitive ratios of $19/24$ and $91/120$ for $\mathcal{\hat{H}}_2$ and $\mathcal{\hat{H}}_3$ respectively, and is optimal again.

\begin{theorem}\label{theorem:hardness-dp-general}
    No algorithm can achieve a competitive ratio strictly larger than $0.7583$ for the (unweighted) oblivious matching problem with the instance $\mathcal{\hat{H}}_3$. 
\end{theorem}

\begin{table}[th]
    \renewcommand{\arraystretch}{1.15}
    \centering
    \caption{Summary of upper bounds for $\mathcal{\hat{H}}_2$ and $\mathcal{\hat{H}}_3$.}
    \label{tab:upper-bounds-results-general}
    \medskip
    \begin{tabular}{lccc}
    \toprule
    Number of vertices on each side $n$ & 2 & 3 \\
    \midrule
    Upper bound ratio  & $0.7917$ & $0.7583$  \\
    Time & \multicolumn{2}{c}{$<$ 1 sec}   \\
    \bottomrule
    \end{tabular}
\end{table}

\section{Discussion}\label{section:discussion}

In this section, we provide some discussions on our algorithms, and make a comparison between our algorithm with the existing ones.

\subsection{Existing Algorithms}\label{subsection:existing-algorithm}

\paragraph{Ranking Algorithm.}
This $(1-\frac{1}{e})$-competitive Online Bipartite Matching algorithm \cite{KVV90,AGKM11,DJK13} can be adopted as a non-adaptive Oblivious Bipartite Matching algorithm for the unweighted and vertex-weighted special cases, with the same competitive ratio of $1-\frac{1}{e}$.

First, consider unweighted matching. 
Ranking starts by sampling a \emph{rank} $y_u \in [0, 1)$ uniformly and independently for each vertex $u \in L$ on the left (the offline vertices in Online Bipartite Matching).
Then, it processes the vertices $v \in R$ on the right (the online vertices) one at a time by an arbitrary order (the arrival order), matching each $v$ to the unmatched neighbor with the smallest rank.
We can imitate how Ranking matches $v$ in the Oblivious Bipartite Matching model by querying $v$'s incident edges in the ascending order of the neighbors' ranks.

Next, consider the problem of (offline-side) vertex-weighted matching, where each vertex $u \in L$ on the left has a non-negative weight $w_u$.
In other words, all edges incident to $u$ have the same weights $w_u$.
Ranking again draws ranks $y_u$ for vertices $u \in L$ on the left, and processes vertices $v \in R$ on the right one at a time.
For each $v \in R$, it queries the edges incident to $v$ in the descending order of $w_u \cdot f(y_u)$, for a discount function $f(y) = 1 - e^{y - 1}$ derived from the analysis. 

\paragraph{Two-Sided Ranking Algorithm.}
Ranking has a competitive ratio larger than $1-\frac{1}{e}$ in the random-order model of Online Bipartite Matching \cite{MY11,KMT11,HTWZ19,JW21,PT25}, where the online vertices $v \in R$ arrive in a random order rather than an arbitrary one. 
It is folklore that the random arrival of online vertices is equivalent to drawing arrival time $y_v \in [0, 1)$ uniformly and independently for online vertices $v \in R$, and the arrival time of a vertex can be seen as its rank.
Indeed, for unweighted matching, the \algoname algorithm (with a strictly decreasing function $g$) is equivalent to the original Ranking algorithm with a random order of online vertices.%
\footnote{The earliest online vertex $v$ (i.e., the one with the smallest rank) would be matched by \algoname with the neighbor $u$ with the smallest rank, because $u$'s other neighbors have larger rank than $v$ and thus the corresponding edges rank lower. Conditioned on this, the second earliest online vertex would be matched with the neighbor other than $u$ with the smallest rank, and so forth.}
Hence, the best known competitive ratio of $0.696$ \cite{MY11} applies.
For vertex-weighted matching, the existing approaches~\cite{HTWZ19,JW21,PT25} treated the online and offline vertices asymmetrically.
These algorithms consider the online vertices in the ascending order of their arrival time/ranks.
For some function $f$ and for each online vertex $v$, they query $v$'s incident edges $(u, v)$ in descending order of $w_u \cdot f(y_u, y_v)$, which equals $v$'s gain if it is matched with $u$.
The bivariate function $f$ in these works cannot be decomposed into the product of two univariate functions of $y_u$ and $y_v$.

\paragraph{Perturbed Greedy.} 
Since we are not constrained by an arrival order of the vertices in the Oblivious Bipartite Matching problem, it is more natural to query the edges $(u, v)$ directly in the descending order of $w_{uv} \cdot f(y_u, y_v)$.
Tang et al.~\cite{TWZ23} gave a simple proof of a $1-\frac{1}{e}$ competitive ratio for edge-weighted Oblivious Bipartite Matching using function $f(y_u, y_v) = 1 - e^{y_u-1}$, where $w_{uv} \cdot f(y_u, y_v)$ is $v$'s gain if it is matched with $u$, similar to the above vertex-weighted approaches.
They also borrowed an idea from a variant of the Ranking algorithm for Fully Online Matching~\cite{HKTWZZ20} and considered $f(y_u, y_v) = 1 - \min \{ y_u, y_v \}$ for general graphs, and let $w_{uv} \cdot f(y_u, y_v)$ be the gain of the vertex with smaller rank between $u$ and $v$.

\bigskip

All approaches above introduce artificial asymmetries among the vertices, even though the Oblivious Bipartite Matching problem is fully symmetric.
Further, they unnecessarily bind two different concepts---the perturbed weight of an edge and the gain of an endpoint.
This binding introduces an unnatural non-monotonicity, as a vertex $u$ may gain more from the match with another vertex $v$ when it has a smaller rank $y_u$, but intuitively $u$ should offer more to $v$ to justify the smaller rank/higher priority.

By contrast, the new approach in this paper restores full symmetry and removes the binding of the two concepts by defining them as quadratic combinations of two univariate functions $g$ and $h$.
The perturbed function $f(y_u, y_v) = g(y_u) \cdot g(y_v)$ is symmetric, while a vertex $u$'s gain $w_{uv} \cdot h(y_u) \cdot g(y_v)$ is only partially bind with the perturbed weights through the common function $g$ and the ranks of the vertices.

\subsection{Deriving the Quadratic Forms from First Principles}
\label{subsection:aximatic}

This subsection shows that the quadratic forms of perturbed weights and guaranteed gains in this paper are the unique choice that satisfy a set of natural properties.

\paragraph{Primal-Dual Perturbed Greedy.}
We first present a general framework capturing \algoname as a special case. 
Let $\rank : [0, 1)^2 \to \Real_{>0}$ and $\share : [0, 1)^2 \to \Real_{>0}$ be two bivariate functions, with the latter function satisfying $\share(x, y) + \share(x, y)\le 1$ for any $x, y \in [0, 1)$.
Consider a Primal-Dual Perturbed Greedy algorithm as follows.
In the beginning, draw a rank $y_v \in [0, 1)$ uniformly and independently for each vertex $v \in V$. 
Then, query the edges $(u, v)$ in the descending order of their perturbed weights defined as
\[
    \hat w(u, v) ~=~ \rank(y_u, y_v)\cdot w_{uv}
    ~.
\]
When the algorithm matches $(u, v)$, it further sets the corresponding dual variables such that
\begin{equation}
\label{eqn:general-share}
\begin{aligned}
    \alpha_u &  ~\geq~ \share(y_u, y_v)\cdot w_{uv} ~, \\
    \alpha_v &  ~\geq~ \share(y_v, y_u)\cdot w_{uv} ~, \\
    \alpha_u + \alpha_v & ~=~ w_{uv}
    ~,
\end{aligned}
\end{equation}
where we will refer to the right-hand-sides of the two inequalities as $u$ and $v$'s \emph{guaranteed gains}.

Similar to the case of \algoname, the restriction that $\rank(x, y)$ and $\share(x, y)$ are strictly positive for any $x, y \in [0, 1)$ is only for the convenience of the arguments, as it prevents the boundary cases of zero perturbed weights and zero gains.
It can be removed with appropriate treatments of the zero-weight cases.

\paragraph{Natural Properties.}
We now present a set of properties and why we consider them natural.
First, it is natural to assume \textbf{Symmetry} of the perturbed weights w.r.t.\ the ranks of its two endpoints, since the Oblivious Bipartite Matching problem is symmetric. 
\begin{align}
    \forall y_u, y_v \in [0, 1):\ \rank(y_u, y_v) = \rank(y_v, y_u).\label{prop:symmetry}
\end{align}

Next, consider the order by which the algorithm queries the incident edges of a vertex $u$, and $u$'s preference over these incident edges based on its guaranteed gain given in \Cref{eqn:general-share}.
It is natural and even desirable to let these two orders be the same, which we will refer to as \textbf{Rank-Share Consistency}.
Formally speaking, we expect that
\begin{align}
    & \forall y_u, y_v, y_{v'} \in [0, 1), \forall w_{uv}, w_{uv'} > 0 : \notag \\[1ex]
    & \quad
    \rank(y_u, y_v) \cdot w_{uv} \ge \rank(y_u, y_{v'}) \cdot w_{uv'} 
    \quad\Leftrightarrow\quad 
    \share(y_u, y_v)\cdot w_{uv} \ge \share(y_u, y_{v'}) \cdot w_{uv'} ~. 
    \label{prop:good-share-1}
\end{align}
That is, the algorithm queries $(u, v)$ before $(u, v')$ if and only if $u$'s guaranteed gain from $(u,v)$ is greater than that from $(u,v')$.
The symmetric claim with the roles of $u, v$ swapped is mathematically identical, and thus, omitted.
Without Rank-Share Consistency, a vertex may gain less at a smaller rank/higher priority, because the algorithm may match it to an earlier edge with a smaller guaranteed gain.

Finally, consider the preference of any vertex $u$ over its incident edges. 
We consider it natural that $u$'s rank $y_u$, which intuitively captures its willingness to share with the other side, shall not affect its preference over the edges.
We call this property \textbf{Preference Consistency}.
Admittedly, this is the least justified property among the three, and we leave it for future research whether it could be replaced by a weaker or alternative property in derivation of the quadratic forms.
\begin{align}
    & \forall y_u, y_u', y_v, y_{v'} \in [0, 1), \forall w_{uv}, w_{uv'} > 0 : \notag \\[1ex]
    & \quad
    \rank(y_u, y_v) \cdot w_{uv} \ge \rank(y_u, y_{v'}) \cdot w_{uv'} \quad\Leftrightarrow\quad \rank(y_u', y_v) \cdot w_{uv} \ge \rank(y_u', y_{v'}) \cdot w_{uv'} ~. 
    \label{prop:consist-pref-1}
\end{align}
In other words, vertex $u$ prefers $(u, v)$ over $(u, v')$ at rank $y_u$ if and only if it has the same preference at another rank $y_u'$.
The symmetric claim with the roles of $u$ and $v$ swapped is mathematically identical, and thus, omitted.

\begin{theorem}
    \label{thm:derive-quadratic-form}
    If functions $\rank$ and $\share$ satisfy Symmetry, Gain-Share Consistency, and Preference Consistency, then there are univariate functions $g, h : [0, 1) \to \Real_{>0}$ such that $\rank(x, y) = g(y_u)\cdot g(y_v)$ and $\share(x, y) = h(x)\cdot g(y)$; and vice versa.
\end{theorem}

The second direction is straightforward and has been verified in the earlier parts of the paper.
The first direction follows from the next two lemmas.

\begin{lemma}\label{lemma:consist-pref}
    If function $\rank$ satisfies Symmetry and Preference Consistency, then there is a function $g : [0, 1) \to \Real_{>0}$ such that $\rank(x, y) = g(x)\cdot g(y)$.
\end{lemma}
\begin{proof}
    It is easy to see that if the first part of \Cref{prop:consist-pref-1} holds with equality, then so does the second part. 
    For any $y_u, y_u', y_v, y_{v'} \in [0, 1)$, let 
    \[
        w_{uv} = \frac{1}{\rank(y_u, y_v)} ~,\quad w_{uv'} = \frac{1}{\rank(y_u, y_{v'})}
        ~,
    \]
    so that the first part of \Cref{prop:consist-pref-1} holds with equality.
    Then, the second part of \Cref{prop:consist-pref-1} with equality is
    \begin{equation}
        \label{eqn:rank-invariant}
        \frac{\rank(y_u', y_v)}{\rank(y_u, y_v)} ~=~
        \frac{\rank(y_u', y_{v'})}{\rank(y_u, y_{v'})}
        ~.
    \end{equation}

    Therefore, for any $x, y\in [0, 1)$, we have
    \begin{align*}
        \rank(x, y) & = \frac{\rank(x, y) \cdot \rank(x,0)}{\rank(0,0)} \cdot \frac{\rank(0, 0)}{\rank(x,0)} \\
        & = \frac{\rank(x, y) \cdot \rank(x,0)}{\rank(0,0)} \cdot \frac{\rank(0,y)}{\rank(x,y)} \tag{by \Cref{eqn:rank-invariant}} \\
        & = \frac{\rank(x, 0)}{\sqrt{\rank(0, 0)}}\cdot \frac{\rank(0, y)}{\sqrt{\rank(0, 0)}} \\
        & = \frac{\rank(x, 0)}{\sqrt{\rank(0, 0)}}\cdot \frac{\rank(y, 0)}{\sqrt{\rank(0, 0)}} 
        ~.
        \tag{by Symmetry}
    \end{align*}

    Hence, the lemma holds with function $g(y) = \frac{\rank(y, 0)}{\sqrt{\rank(0, 0)}}$. 
\end{proof}

\begin{lemma}\label{lemma:good-share}
    If Rank-Share Consisntency holds and $\rank(x, y) = g(x)\cdot g(y)$ for some $g: [0, 1) \to \mathbb R_{> 0}$, then there is a function $h: [0, 1)\to \Real_{> 0}$ such that $\share(x, y) = h(x)\cdot g(y)$.
\end{lemma}
\begin{proof}
    It is easy to see that if the first part of \Cref{prop:good-share-1} holds with equality, then so does the second part. 
    For any $y_u, y_v, y_{v'} \in [0, 1)$, let
    \[
        w_{uv} = \frac{1}{\rank(y_u, y_v)}
        ~,\quad
        w_{uv'} = \frac{1}{\rank(y_u, y_{v'})}
        ~,
    \]
    so that the first part of \Cref{prop:good-share-1} holds with equality.
    Then, the second part with equality is
    \begin{equation}
        \label{eqn:gain-share-consistency}
        \frac{\share(y_u, y_v)}{\rank(y_u, y_v)} 
        ~=~ 
        \frac{\share(y_u, y_{v'})}{\rank(y_u, y_{v'})}
        ~.
    \end{equation}

    Therefore, for any $x, y\in [0, 1)$, we have
    \begin{align*}
        \share(x, y) & = \frac{\share(x, y) \cdot \share(x, 0)}{\rank(x, 0)} \cdot \frac{\rank(x, 0)}{\share(x, 0)} \\
        & =  \frac{\share(x, y) \cdot \share(x, 0)}{\rank(x, 0)} \cdot \frac{\rank(x, y)}{\share(x, y)} 
        \tag{by \Cref{eqn:gain-share-consistency}} \\
        & = \frac{\share(x, 0) \cdot \rank(x, y)}{\rank(x, 0)}\\
        & = \frac{\share(x, 0)}{g(0)} \cdot g(y) \tag{by $\rank(x,y) = g(x) g(y)$}
        ~.%
    \end{align*}

    Hence, the lemma holds with $h(y) = \frac{\share(x, 0)}{g(0)}$.
\end{proof}

\section{Conclusion and Open Questions}

In this work, we propose a novel polynomial-time algorithm named \algoname, for the Edge-weighted Oblivious Bipartite Matching problem, and prove that it is at least $0.659$-competitive using randomized primal-dual analysis.
The algorithm is parameterized by two functions $g$ and $h$ that define the perturbed weight of the edges and the dual variable assignments, and satisfies several natural and desirable properties, including Symmetry, Rank-share Consistency, and Preference Consistency.
These properties allow us to extend existing analysis frameworks for the unweighted or vertex-weighted version of the problem to the edge-weighted case, and prove the first competitive ratio beating $1-1/e$.
We also show upper bounds of $0.796$ and $0.758$ for the competitive ratios of any (randomized, adaptive) algorithms for the unweighted Oblivious Matching problem on bipartite and general graphs, respectively.

A few interesting problems are left open.
As mentioned by Derakhshan and Farhad~\cite{DF23}, a natural open question is whether there is a separation in the best possible competitive ratios for the model of Oblivious Matching and that of the Query-Commit Matching.
One possible direction towards answering this open question is to further improve the upper bounds for the Oblivious Matching problem that are given by the hard instances $\mathcal{H}_6$ and $\mathcal{\hat{H}}_3$.
However, we remark that studying even larger hard instances with similar graph structures either leads to an improvement of less than $0.001$ ($\mathcal{H}_7$) or does not lead to an improvement at all ($\hat {\mathcal{H}}_4$).
Finally, another natural question is whether \algoname can be applied to other related problems, e.g., Edge-weighted Oblivious Matching in general graphs, and give a good competitive ratio.

\newpage

\bibliography{ref}

\newcommand{\etalchar}[1]{$^{#1}$}
\begin{thebibliography}{AGKM11}

\bibitem[ABG{\etalchar{+}}20]{ABGSSX20}
Marek Adamczyk, Brian Brubach, Fabrizio Grandoni, Karthik~Abinav Sankararaman, Aravind Srinivasan, and Pan Xu.
\newblock Improved approximation algorithms for stochastic-matching problems.
\newblock {\em CoRR}, abs/2010.08142, 2020.

\bibitem[Ada11]{A11}
Marek Adamczyk.
\newblock Improved analysis of the greedy algorithm for stochastic matching.
\newblock {\em Inf. Process. Lett.}, 111(15):731--737, 2011.

\bibitem[ADFS95]{ADFS95}
Jonathan Aronson, Martin~E. Dyer, Alan~M. Frieze, and Stephen Suen.
\newblock Randomized greedy matching {II}.
\newblock {\em Random Struct. Algorithms}, 6(1):55--74, 1995.

\bibitem[AGKM11]{AGKM11}
Gagan Aggarwal, Gagan Goel, Chinmay Karande, and Aranyak Mehta.
\newblock Online vertex-weighted bipartite matching and single-bid budgeted allocations.
\newblock In {\em Proceedings of the 22nd Annual ACM-SIAM Symposium on Discrete Algorithms}, pages 1253--1264. SIAM, 2011.

\bibitem[BC21]{BC21}
Guy Blanc and Moses Charikar.
\newblock Multiway online correlated selection.
\newblock In {\em {FOCS}}, pages 1277--1284. {IEEE}, 2021.

\bibitem[CCW18]{CCW18}
T.{-}H.~Hubert Chan, Fei Chen, and Xiaowei Wu.
\newblock Analyzing node-weighted oblivious matching problem via continuous {LP} with jump discontinuity.
\newblock {\em {ACM} Trans. Algorithms}, 14(2):12:1--12:25, 2018.

\bibitem[CCWZ18]{CCWZ18}
T.{-}H.~Hubert Chan, Fei Chen, Xiaowei Wu, and Zhichao Zhao.
\newblock Ranking on arbitrary graphs: Rematch via continuous linear programming.
\newblock {\em {SIAM} J. Comput.}, 47(4):1529--1546, 2018.

\bibitem[CHLT25]{CHLT25}
Ziyun Chen, Zhiyi Huang, Dongchen Li, and Zhihao~Gavin Tang.
\newblock Prophet secretary and matching: the significance of the largest item.
\newblock In {\em {SODA}}, pages 1371--1401. {SIAM}, 2025.

\bibitem[CHS24]{CHS24}
Ziyun Chen, Zhiyi Huang, and Enze Sun.
\newblock Stochastic online correlated selection.
\newblock In {\em {FOCS}}, pages 2275--2294. {IEEE}, 2024.

\bibitem[CIK{\etalchar{+}}09]{CIKMR09}
Ning Chen, Nicole Immorlica, Anna~R. Karlin, Mohammad Mahdian, and Atri Rudra.
\newblock Approximating matches made in heaven.
\newblock In {\em {ICALP} {(1)}}, volume 5555 of {\em Lecture Notes in Computer Science}, pages 266--278. Springer, 2009.

\bibitem[CTT12]{CTT12}
Kevin~P. Costello, Prasad Tetali, and Pushkar Tripathi.
\newblock Stochastic matching with commitment.
\newblock In {\em {ICALP} {(1)}}, volume 7391 of {\em Lecture Notes in Computer Science}, pages 822--833. Springer, 2012.

\bibitem[DF91]{DF91}
Martin~E. Dyer and Alan~M. Frieze.
\newblock Randomized greedy matching.
\newblock {\em Random Struct. Algorithms}, 2(1):29--46, 1991.

\bibitem[DF23]{DF23}
Mahsa Derakhshan and Alireza Farhadi.
\newblock Beating {(1} - 1/e)-approximation for weighted stochastic matching.
\newblock In {\em {SODA}}, pages 1931--1961. {SIAM}, 2023.

\bibitem[DJK13]{DJK13}
Nikhil~R. Devanur, Kamal Jain, and Robert~D. Kleinberg.
\newblock Randomized primal-dual analysis of {RANKING} for online bipartite matching.
\newblock In {\em {SODA}}, pages 101--107. {SIAM}, 2013.

\bibitem[FHTZ22]{FHTZ22}
Matthew Fahrbach, Zhiyi Huang, Runzhou Tao, and Morteza Zadimoghaddam.
\newblock Edge-weighted online bipartite matching.
\newblock {\em J. {ACM}}, 69(6):45:1--45:35, 2022.

\bibitem[FTW{\etalchar{+}}21]{FTWWZ21}
Hu~Fu, Zhihao~Gavin Tang, Hongxun Wu, Jinzhao Wu, and Qianfan Zhang.
\newblock Random order vertex arrival contention resolution schemes for matching, with applications.
\newblock In {\em {ICALP}}, volume 198 of {\em LIPIcs}, pages 68:1--68:20. Schloss Dagstuhl - Leibniz-Zentrum f{\"{u}}r Informatik, 2021.

\bibitem[GKS19]{GKS19}
Buddhima Gamlath, Sagar Kale, and Ola Svensson.
\newblock Beating greedy for stochastic bipartite matching.
\newblock In {\em {SODA}}, pages 2841--2854. {SIAM}, 2019.

\bibitem[GT12]{GT12}
Gagan Goel and Pushkar Tripathi.
\newblock Matching with our eyes closed.
\newblock In {\em {FOCS}}, pages 718--727. {IEEE} Computer Society, 2012.

\bibitem[{Gur}25]{Gurobi}
{Gurobi Optimization, LLC}.
\newblock {Gurobi Optimizer Reference Manual}, 2025.

\bibitem[HKT{\etalchar{+}}20]{HKTWZZ20}
Zhiyi Huang, Ning Kang, Zhihao~Gavin Tang, Xiaowei Wu, Yuhao Zhang, and Xue Zhu.
\newblock Fully online matching.
\newblock {\em J. {ACM}}, 67(3):17:1--17:25, 2020.

\bibitem[HTWZ19]{HTWZ19}
Zhiyi Huang, Zhihao~Gavin Tang, Xiaowei Wu, and Yuhao Zhang.
\newblock Online vertex-weighted bipartite matching: Beating 1-1/\emph{e} with random arrivals.
\newblock {\em {ACM} Trans. Algorithms}, 15(3):38:1--38:15, 2019.

\bibitem[HTWZ20]{HTWZ20}
Zhiyi Huang, Zhihao~Gavin Tang, Xiaowei Wu, and Yuhao Zhang.
\newblock Fully online matching {II:} beating ranking and water-filling.
\newblock In {\em {FOCS}}, pages 1380--1391. {IEEE}, 2020.

\bibitem[JW21]{JW21}
Billy Jin and David~P. Williamson.
\newblock Improved analysis of {RANKING} for online vertex-weighted bipartite matching in the random order model.
\newblock In {\em {WINE}}, volume 13112 of {\em Lecture Notes in Computer Science}, pages 207--225. Springer, 2021.

\bibitem[KMT11]{KMT11}
Chinmay Karande, Aranyak Mehta, and Pushkar Tripathi.
\newblock Online bipartite matching with unknown distributions.
\newblock In {\em {STOC}}, pages 587--596. {ACM}, 2011.

\bibitem[KVV90]{KVV90}
Richard~M. Karp, Umesh~V. Vazirani, and Vijay~V. Vazirani.
\newblock An optimal algorithm for on-line bipartite matching.
\newblock In Harriet Ortiz, editor, {\em Proceedings of the 22nd Annual {ACM} Symposium on Theory of Computing, May 13-17, 1990, Baltimore, Maryland, {USA}}, pages 352--358. {ACM}, 1990.

\bibitem[LW21]{LW21}
Roie Levin and David Wajc.
\newblock Streaming submodular matching meets the primal-dual method.
\newblock In {\em {SODA}}, pages 1914--1933. {SIAM}, 2021.

\bibitem[MY11]{MY11}
Mohammad Mahdian and Qiqi Yan.
\newblock Online bipartite matching with random arrivals: an approach based on strongly factor-revealing lps.
\newblock In {\em {STOC}}, pages 597--606. {ACM}, 2011.

\bibitem[PRSW22]{PRSW22}
Tristan Pollner, Mohammad Roghani, Amin Saberi, and David Wajc.
\newblock Improved online contention resolution for matchings and applications to the gig economy.
\newblock In {\em {EC}}, pages 321--322. {ACM}, 2022.

\bibitem[PS12]{PS12}
Matthias Poloczek and Mario Szegedy.
\newblock Randomized greedy algorithms for the maximum matching problem with new analysis.
\newblock In {\em {FOCS}}, pages 708--717. {IEEE} Computer Society, 2012.

\bibitem[PT25]{PT25}
Bo~Peng and Zhihao~Gavin Tang.
\newblock Revisiting ranking for online bipartite matching with random arrivals: the primal-dual analysis, 2025.

\bibitem[TWZ23]{TWZ23}
Zhihao~Gavin Tang, Xiaowei Wu, and Yuhao Zhang.
\newblock Toward a better understanding of randomized greedy matching.
\newblock {\em J. {ACM}}, 70(6):39:1--39:32, 2023.

\bibitem[TZ24]{TZ24}
Zhihao~Gavin Tang and Yuhao Zhang.
\newblock Improved bounds for fractional online matching problems.
\newblock In {\em {EC}}, pages 279--307. {ACM}, 2024.

\end{thebibliography}
\bibliographystyle{alpha}

\newpage

\appendix
\section{Missing Proofs from \texorpdfstring{\Cref{sec:quadratic-ranking}}{}}

\subsection{Proof for \texorpdfstring{\Cref{lemma:universal-lower-bound}}{}}
\label{app-subsec:proof-universal-lower-bound}

\begin{lemma}
    \label{app-lemma:non-decreasing-right-continuous}
    Consider a non-decreasing right-continuous function $f : [0, 1] \to \Real_{\ge 0}$.
    For any $x, y\in [0, 1]$, if $f^{-1}(y) < 1$, then $x > f^{-1}(y) \implies f(x) > y$; $x \le f^{-1}(y) \implies f(x) \le y$.
\end{lemma}
\begin{proof}
    If $f^{-1}(y) < 1$, then $S = \{z\in[0, 1]: f(z) > y\}$ is non-empty. 
    Since $f$ is right-continuous, we have $f^{-1}(y) = \inf S$.  
    If $x > f^{-1}(y)$, then $x> \inf S$. Since $f$ is non-decreasing, we have $x\in S$, which implies $f(x) > y$. 
    Similarly, if $x \le f^{-1}(y)$, then $x\le \inf S$.
    Therefore $x \notin S$, which means $f(x) \leq y$.
\end{proof}

\begin{lemma}[\Cref{lemma:universal-lower-bound} restated]\label{app:lemma:universal-lower-bound}
    We have
    \begin{align*}
        \frac{\E[\alpha_u + \alpha_v]}{w_{uv}} \geq \int_0^1 \underbrace{\Big( \theta(y_v) - \beta^{-1}(y_v) \Big)^+}_{\text{\rm shaded area, \Cref{lemma:extra-gain}}} \d y_v & + \int_0^1 \underbrace{\Big(1 - \big(\theta(y_v) - \beta^{-1}(y_v)\big)^+ \Big) \cdot h(y_v) \cdot  g(\theta(y_v))}_{\text{\rm non-shaded area, \Cref{lemma:basic-gain} for $\alpha_v$}} \,\d y_v \\
        & + \int_0^1 \underbrace{\Big(1 - \big(\beta(y_u) - \theta^{-1}(y_u)\big)^+\Big)\cdot h(y_u) \cdot g(\beta(y_u))}_{\text{\rm non-shaded area, \Cref{lemma:basic-gain} for $\alpha_u$}}  \,\d y_u
        ~,
    \end{align*}
    where $x^+$ denotes $\max\{x, 0\}$.
\end{lemma}

\begin{proof}
    We break $\E[\alpha_u + \alpha_v]$ into three parts and lower bound each of them (see Figure~\ref{fig:universal_lower_bound}):
    \begin{align*}
        \E[\alpha_u + \alpha_v] & = \E\Big[(\alpha_u + \alpha_v)\cdot \bm{1}\Big( (y_v < \beta(y_u)) \wedge (y_u < \theta(y_v)) \Big)\Big] \\
        & + \E\Big[\alpha_v \cdot \bm{1}\Big( (y_v \ge \beta(y_u)) \vee (y_u \ge \theta(y_v)) \Big)\Big] \\
        & + \E\Big[\alpha_u \cdot \bm{1}\Big( (y_v \ge \beta(y_u)) \vee (y_u \ge \theta(y_v)) \Big)\Big].
    \end{align*}

    Here, $\bm{1}(A)\in \{0,1\}$ is the indicator function for event $A$.
    For the first term, we consider any fixed $y_v \in [0, 1)$. 
    If $\beta^{-1}(y_v) < \theta(y_v)$ (which implies $\beta^{-1}(y_v) < 1$), then for $y_u \in (\beta^{-1}(y_v), \theta(y_v))$, we have both $y_v < \beta(y_u)$ ($\beta^{-1}(y_v) < 1$ and $y_u < \theta(y_v)$ by  \Cref{app-lemma:non-decreasing-right-continuous}. 
    
    Then by \Cref{lemma:extra-gain} we have $\alpha_u + \alpha_v = w_{uv}$.
    Therefore, we have,
    \begin{equation*}
        \E\Big[(\alpha_u + \alpha_v)\cdot \bm{1}\Big( (y_v < \beta(y_u)) \wedge (y_u < \theta(y_v)) \Big)\Big] \geq w_{uv}\cdot\int_0^1 \Big( \theta(y_v) - \beta^{-1}(y_v) \Big)^+ \d y_v.
    \end{equation*}

    For the second term, we consider any fixed $y_v \in [0, 1]$.
    If $\theta(y_v) = 1$, we have $g(\theta(y_v)) = 0$. 
    Thus we only need to consider the case when $\theta(y_v) < 1$.
    If $\beta^{-1}(y_v) < \theta(y_v) < 1$, then for $y_u \notin (\beta^{-1}(y_v), \theta(y_v))$, we have either $y_u \ge \theta(y_v)$ or $y_v \ge \beta(y_u)$ (the second is by \Cref{app-lemma:non-decreasing-right-continuous}).
    Then by \Cref{lemma:basic-gain} we have $\alpha_v \geq h(y_v)\cdot g(\theta(y_v))\cdot w_{uv}$. 
    Therefore we have
    \begin{equation*}
        \E\Big[\alpha_v \cdot \bm{1}\Big( (y_v \ge \beta(y_u)) \vee (y_u \ge \theta(y_v)) \Big)\Big] \ge \int_0^1 \Big(1 - \big(\theta(y_v) - \beta^{-1}(y_v)\big)^+ \Big) \cdot h(y_v) \cdot  g(\theta(y_v)) \d y_v.
    \end{equation*}

    Symmetrically, we can show that
    \begin{equation*}
        \E\Big[\alpha_u \cdot \bm{1}\Big( (y_v \ge \beta(y_u)) \vee (y_u \ge \theta(y_v)) \Big)\Big] \le \int_0^1 \Big(1 - \big(\beta(y_u) - \theta^{-1}(y_u)\big)^+\Big)\cdot h(y_u) \cdot g(\beta(y_u)) \d y_u.
    \end{equation*}

    Combining the three lower bounds and rearranging the inequality yields the lemma.
\end{proof}

\subsection{Proof for Lemma \ref{lemma:add-neighbor}}\label{app:proof4lemma:add-neighbor}

\begin{lemma}[\Cref{lemma:add-neighbor} restated]
    For any ranks $y_u$ and $y_v$, vertex $u$ prefers $M(y_u, y_v)$ over $M(y_u, \bot)$.
    Symmetrically, vertex $v$ prefers $M(y_u, y_v)$ over $M(\bot, y_v)$.
\end{lemma}
\begin{proof}
    By symmetric, it suffices to prove the statement that $v$ prefers $M(y_u, y_v)$ over $M(\bot, y_v)$.
    
    With the ranks of vertices other than $u$ fixed, we consider two runs the algorithms with $u$'s rank being $y_u$ and $u$ removed (i.e., $y_u = \bot$) respectively.
    We will process the edges in the two runs by the descending order of their perturbed weights.
    Notice that the difference in these two runs is that there are more edges in the first run, which are the edges incident to $u$.
    For convenience of comparison between the two runs, we add a dummy copy of $u$ to the second run with rank $y_u$, but mark it as matched at the beginning.
    Suppose that $u\in L$. We use mathematical induction to show that at any moment when the same edge is considered in the two runs, 
    \begin{itemize}
        \item the set of unmatched vertices in $L$ in the first run is a superset of those in the second;
        \item the set of matched vertices in $R$ in the first run is a superset of those in the second.
    \end{itemize}

    Intuitively speaking, due to the existence of $u$, in the first run, there will be more unmatched vertices in $L$, and the vertices in $R$ are more likely to get matched early.
    The base case holds trivially when the algorithm queries the first edge.
    Now suppose that the statement holds before edge $(p,q)\in L\times R$ is considered, we show that the statement still holds afterwards, which is trivially true if the query outcome is the same in both runs.
    Now suppose that $(p,q)$ is included in the first run but not in the second.
    By induction hypothesis, when $(p,q)$ is considered in the second run, $q$ is unmatched, which implies that $p$ is already matched, and the statement follows.
    Similarly, i0f $(p,q)$ is included in the second run but not in the first, then $p$ must be unmatched in the first run, which implies that $q$ is already matched, and the statement follows.

    Next we show that $v$ prefers $M(y_u, y_v)$ over $M(\bot, y_v)$.
    
    If $v$ is unmatched in the second run, the lemma holds trivially.
    Otherwise, suppose $v$ is matched to $p$ in the second run.
    By the previous argument, when $(p,v)$ is considered in the first run, $p$ must be unmatched.
    Therefore, either $v$ is already matched, or $v$ would be matched to $p$.
    In either case, the statement follows.
\end{proof}

\subsection{Proof for \texorpdfstring{\Cref{lemma:general-h-satisfy-constraint}}{}}
\label{app:proof-general-h-satisfy-constraint}

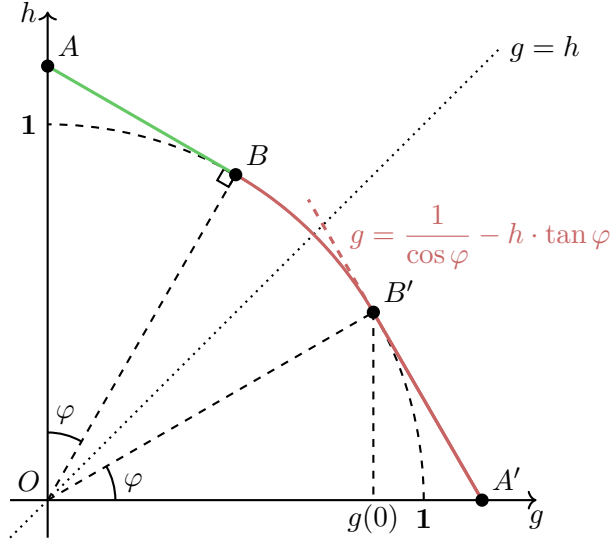
\begin{figure}[!ht]
    \centering
    \begin{tikzpicture}[scale = 5]
        \draw[->, thick] (-0.1,0) -- (1.3,0) node[below] {$g$};
        \draw[->, thick] (0,-0.1) -- (0,1.3) node[left] {$h$};

        \draw[black, thick, dashed] (1, 0) arc (0:90:1);

        \node at (-0.05, 1) {$\mathbf{1}$};
        \node at (1, -0.05) {$\mathbf{1}$};
        \node at (-0.05, 0.05) {$O$};
        \draw[-, thick, dashed] (0.866, 0.5) -- (0.866, 0);
        \node at (0.866, -0.05) {$g(0)$};

        \coordinate (A) at (0.5,0.866);
        \draw[black, thick, rotate=60] (A) rectangle ($(A) + (-0.035,0.035)$);
        
        \draw[myGreen, very thick] (0,1.155) -- (0.5,0.866);  
        \draw[myRed, very thick]   (0.866, 0.5) -- (1.155, 0);
        \draw[myRed, very thick]   (0.866, 0.5) -- (1.155, 0); 

        \draw[black, dotted, thick] (-0.1, -0.1) -- (1.2, 1.2) node[right] {$g = h$};
        
        \draw[myRed, very thick] (0.866, 0.5) arc (30:60:1);
        \draw[myRed, very thick, dashed]  (0.866, 0.5) -- (0.693, 0.8); \node[myRed] at (1.15, 0.7) {$g = \displaystyle\frac{1}{\cos \varphi} - h\cdot\tan \varphi$};

        \fill[black] (0.5,0.866) circle (0.5pt) node[above right] {$B$};

        \fill[black] (0.866, 0.5) circle (0.5pt) node[above right] {$B'$};

        \fill[black] (0, 1.155) circle (0.5pt) node[above right] {$A$};

        \fill[black] (1.155, 0) circle (0.5pt) node[above right] {$A'$};

        \draw[black, thick] (0,0) ++(0:0.18cm) arc (0:30:0.18cm) node[midway, right] {$\varphi$};  
        \draw[black, thick] (0,0) ++(60:0.18cm) arc (60:90:0.18cm) node[midway, above] {$\varphi$};

        \draw[black, thick, dashed]  (0, 0) -- (0.866, 0.5);
        \draw[black, thick, dashed]  (0, 0) -- (0.5,0.866);
    \end{tikzpicture}

    \caption{Illustration for $\vec u\cdot \vec v\le 1$. When $g(x)\in [0, \sin \varphi)$, the point $(g(x), h(y))$ is on the green linear segment. By symmetry, the point $(h(y), g(y))$ is on the red curve.}
    \label{fig:one-vector-on-segment}
\end{figure}

\begin{lemma}[\Cref{lemma:general-h-satisfy-constraint} restated]\label{app-lemma:general-h-satisfy-constraint}
    For any function $g : [0, 1] \to [0, \cos \varphi]$, and 
    \[
        h(y) = \begin{cases}
                \dfrac{1}{\cos \varphi} - g(y) \cdot \tan \varphi & \mbox{if $0\le g(y)< \sin \varphi$\,;} \\[2ex]
                \sqrt{1 - g(y)^2} & \mbox{if $\sin \varphi\le g(y) \le \cos \varphi$\,,}
        \end{cases}
    \]
    Constraint \eqref{eqn:budget-balanced} holds for any $x, y \in [0, 1]$.
\end{lemma}
\begin{proof}
    Fix arbitrary $x, y\in [0, 1]$. If both $g(x)$ and $g(y)$ are in $[\sin \varphi, \cos \varphi]$, then we have
    \[
        h(x)g(y) + h(y)g(x) \le \frac{1 - g(x)^2 + g(y)^2}{2} + \frac{1 - g(y)^2 + g(x)^2}{2} = 1,
    \]
    where the inequality holds by the inequality of arithmetic and geometric means.

    Otherwise, suppose $g(x)\in [0, \sin \varphi)$. We consider the inner product of vectors $\vec u = (g(x), h(x))$ and $\vec v = (h(y), g(y))$. As shown in \Cref{fig:one-vector-on-segment}, by symmetry, we have
    \[
        g(y) \le \frac{1}{\cos \varphi} - h(y)\cdot \tan \varphi.
    \]
    Let
    \[
        A = \left[\begin{matrix}
            \sin \varphi & \cos \varphi \\
            \cos \varphi & -\sin \varphi
        \end{matrix}\right].
    \]
    We have $\vec u \cdot \vec v = A \vec u \cdot A\vec v$ since $\lvert A\rvert = 1$.
    Thus,
    \begin{align*}
        \vec u \cdot \vec v = \left[\begin{matrix}
            g(x) \\
            h(x) 
        \end{matrix}\right]
        \cdot
        \left[\begin{matrix}
            h(y) \\
            g(y) 
        \end{matrix}\right] &  = 
        \left[\begin{matrix}
            g(x)\sin \varphi + h(x)\cos\varphi \\
            g(x)\cos \varphi - h(x)\sin\varphi 
        \end{matrix}\right]
        \cdot
        \left[\begin{matrix}
            h(y) \sin \varphi + g(y)\cos\varphi \\
            h(y) \cos \varphi - g(y)\sin\varphi
        \end{matrix}\right].
    \end{align*}

    To prove $\vec u \cdot \vec v \le 1$, we bound each term in the two vectors.
    \begin{align*}
        g(x)\sin \varphi + h(x)\cos\varphi &= g(x)\sin \varphi + 1 - g(x)\sin \varphi = 1; \\
        h(y) \sin \varphi + g(y)\cos\varphi &\le h(y)\sin \varphi + 1 - h(y)\sin \varphi = 1; \\
        g(x)\cos \varphi - h(x)\sin\varphi & = g(x)\cos \varphi - \tan \varphi + g(x)\frac{\sin^2 \varphi}{\cos \varphi}
         = \frac{g(x) - \cos\varphi}{\cos \varphi} \le 0;\\ 
        h(y) \cos \varphi - g(y)\sin\varphi & \ge h(y)\cos \varphi - \tan \varphi + h(y)\frac{\sin^2 \varphi}{\cos \varphi} 
         = \frac{h(y) - \sin\varphi}{\cos \varphi} \ge 0. 
    \end{align*}

    Therefore, $\vec u \cdot \vec v \le 1 \times 1 + 0 = 1$.
\end{proof}

\section{Analytical Competitive Analysis}
\label{app:analytical-competitive-analysis}

In this section we provide an analytical lower bound of $0.63245$ for $\E[\alpha_u + \alpha_v]$ for some fixed $g$ and $h$, thus proving Theorem~\ref{theorem:analytical-lower-bound}.
We first introduce two variables $\tau, \gamma \in [0,1]$ defined by $\theta$ and $\beta$, and prove a lower bound for $\E[\alpha_u + \alpha_v]$ that depends only on $\tau,\gamma$ and the functions $g$ and $h$.

\begin{definition} 
    Let $\tau := \sup\{ y \in [0,1] : \theta(y) < 1 \}$ and $\gamma := \sup\{ y \in [0,1]: \beta(y) < 1 \}$.
\end{definition}

With a constraint on $g$ and $h$, we have the following lower bound.

\begin{lemma}[Lower Bound for Analytical Analysis] \label{lemma:analytical_lower_bound}
    Fix a differentiable function g. Under the condition that $h(1)\cdot ( g(y) - g'(y)) \leq 1$ for all $y\in [0,1]$,
    we have
    \begin{align*}
        \frac{\E[\alpha_u + \alpha_v]}{w_{uv}} & \geq (1 - \tau)(1 - \gamma) 
        + (1-\tau)\cdot g(\tau)\int_0^\gamma h(x) \d x \\
        & \quad + \int_0^\tau \min_{t \leq \gamma}\left\{ h(y)\cdot g(t) + g(y) \int_0^t h(x) \d x 
        + g(\tau) \int_t^\gamma h(x) \d x \right\} \d y.
    \end{align*}
\end{lemma}

\subsection{Proof of Lemma~\ref{lemma:analytical_lower_bound}}

We break $\E[\alpha_u + \alpha_v]$ into three parts and lower bound each of them:
\begin{align*}
    \E[\alpha_u + \alpha_v] & = \E[(\alpha_u + \alpha_v) \cdot \mathbf{1}(y_v > \tau, y_u > \gamma)] \\
    & + \E[(\alpha_u +\alpha_v) \cdot \mathbf{1}(y_v > \tau, y_u < \gamma)] \\
    & + \E[(\alpha_u +\alpha_v) \cdot \mathbf{1}(y_v < \tau)].
\end{align*}    
    
For the first term, observe that when $y_v\in (\tau,1)$ and $y_u\in (\gamma,1)$, we have both $y_u < \theta(y_v) = 1$ and $y_v < \beta(y_u) = 1$.
Then by Lemma~\ref{lemma:extra-gain}, we have 
\begin{equation}
    \E[(\alpha_u + \alpha_v) \cdot \mathbf{1}(y_v > \tau, y_u > \gamma)] = (1-\tau)(1-\gamma)\cdot w_{uv}. \label{equation:analytical_lower_bound_1}
\end{equation}

Now we lower bound the second term. 
Fix any $x < \gamma$ and consider the case when $y_u = x$. By definition of $\gamma$ we have $\beta(x) < 1$. 
Hence by Lemma~\ref{lemma:basic-gain}, for all $y_v\in (\tau, 1)$ we have $\alpha_u \ge h(x)\cdot g(\beta(x)) \cdot w_{uv}$. 
Moreover, if $\beta(x) > \tau$, then when $y_v\in (\tau, \beta(x))$ we have $y_v < \beta(x)$ and $y_u = x < \gamma \le \theta(y_v) = 1$. 
Then by Lemma~\ref{lemma:extra-gain} we have $\alpha_u + \alpha_v = w_{uv}$.
Therefore we have
\begin{align*}
    &\ \E[(\alpha_u + \alpha_v)\cdot \mathbf{1}(y_v > \tau, y_u < \gamma)]
    = \int_0^\gamma \int_\tau^1 (\alpha_u + \alpha_v) \d y_v \d y_u \\
    & \ge \int_0^\gamma \left( \int_{\tau}^{\max\{\tau, \beta(y_u)\}} (\alpha_u + \alpha_v) \d y_v + \int_{\max\{ \tau, \beta(y_u) \} }^1 \alpha_u \d y_v \right) \d y_u \\
    & \ge \int_0^\gamma \left( \int_{\tau}^{\max\{\tau, \beta(y_u)\}} w_{uv} \d y_v + \int_{\max\{ \tau, \beta(y_u) \} }^1 h(y_u)g(\beta(y_u))w_{uv}\d y_v \right) \d y_u \\
    & \ge w_{uv} \cdot \int_0^\gamma \Big( \big(1 - \max\{\tau, \beta(y_u) \} \big)h(y_u)g(\beta(y_u)) + (\beta(y_u) - \tau)^+  \Big)\d y_u.
\end{align*}

Next, we give a lower bound for the third term. Fix any $y < \tau$ and consider the case when $y_v = y$. By definition of $\tau$ we have $\theta(y) <  1$. Then, by Lemma~\ref{lemma:basic-gain}, for all $y_u \in (0, 1)$, we have $\alpha_v \ge h(y)g(\theta(y))w_{uv}$. 
Moreover, if $\theta(y) > \gamma$, then when $y_u \in (\gamma, \theta(y))$\footnote{The statement holds for all $y_u \in (\beta^{-1}(y), \theta(y))$. Here we only consider $y_u \in (\gamma, \theta(y))$ to get rid of the inverse function and simplify the lower bound.}, we have $\alpha_u + \alpha_v = w_{uv}$.

In addition, for any $y_u \in (0, \min\{\theta(y), \gamma\})$, 
\begin{itemize}
    \item if $\beta(y_u) > y$, then together with $y_u < \theta(y)$, we have that $u$ is matched with $v$ and $\alpha_u \ge h(y_u)g(y)w_{uv}$ by Lemma~\ref{lemma:extra-gain}; 
    \item if $\beta(y_u) \le y$, then by Lemma~\ref{lemma:basic-gain}, $\alpha_u \ge h(y_u)g(\beta(y_u))w_{uv}\ge h(y_u)g(y)w_{uv}$ (the last inequality holds since $g$ is non-increasing).
\end{itemize}

Thus, in both cases $\alpha_u \ge  h(y_u)g(y)w_{uv}$. Moreover, If $\theta(y) < \gamma$, then when $y_u \in (\theta(y), \gamma)$, by Lemma~\ref{lemma:basic-gain} we have $\alpha_u \ge h(y_u)g(\beta(y_u))w_{uv}$.
Let $\bar{\theta}(y) = \min\{ \theta(y), \gamma \}$ and $\hat{\theta}(y) = \max\{ \theta(y), \gamma \}$. 
Summarizing the above arguments, we have
\begin{align*}
    & \quad \frac{1}{w_{uv}} \cdot \E[(\alpha_u + \alpha_v)\cdot \mathbf{1}(y_v < \tau)]
    = \frac{1}{w_{uv}} \cdot \int_0^{\tau} \int_0^1 (\alpha_u + \alpha_v) \d y_u \d y_v \\
    & \ge \frac{1}{w_{uv}} \cdot \int_0^\tau \left( \int_0^\gamma \alpha_v \d y_u + \int_0^{\bar{\theta}(y_v)} \alpha_u\d y_u + \int_{\bar{\theta}(y_v)}^{\gamma} \alpha_u\d y_u  + \int_\gamma^{\hat{\theta}(y_v)} (\alpha_u + \alpha_v)\d y_u + \int_{\hat{\theta}(y_v)}^1 \alpha_v\d y_u  \right) \d y_v\\
    & \ge \int_0^\tau \Bigg( \int_0^\gamma h(y_v)g(\theta(y_v)) \d y_u + \int_0^{\bar{\theta}(y_v)} h(y_u)g(y_v)\d y_u + \int_{\bar{\theta}(y_v)}^{\gamma} h(y_u)g(\beta(y_u))\d y_u \\
    &\qquad\qquad\qquad + \int_\gamma^{\hat{\theta}(y_v)} 1\d y_u + \int_{\hat{\theta}(y_v)}^1 h(y_v)g(\theta(y_v)) \d y_u \Bigg) \d y_v\\
    & = \int_0^\tau \Bigg( ( 1 - (\theta(y_v) - \gamma)^+ )h(y_v)g(\theta(y_v)) + (\theta(y_v) - \gamma)^+ \\
    & \qquad \qquad \qquad  + \int_0^{\bar{\theta}(y_v)} h(y_u)g(y_v)\d y_u + \int_{\bar{\theta}(y_v)}^\gamma h(y_u)g(\beta(y_u)) \d y_u \Bigg) \d y_v.
\end{align*}

Next, we put the lower bounds of $\E[(\alpha_u +\alpha_v) \cdot \mathbf{1}(y_v > \tau, y_u < \gamma)]$ and $\E[(\alpha_u + \alpha_v)\cdot \mathbf{1}(y_v < \tau)]$ together, and show that the minimum is attained when $\beta(y_u) = \tau$ and $\theta(y_v) \leq \gamma$.
Note that
\begin{align}
    &\ \frac{1}{w_{uv}}\cdot \E[(\alpha_u +\alpha_v) \cdot \mathbf{1}(y_v > \tau, y_u < \gamma)] + \frac{1}{w_{uv}}\cdot \E[(\alpha_u + \alpha_v)\cdot \mathbf{1}(y_v < \tau)] \nonumber \\ 
    & \ge \int_0^\gamma \Big( \big(1 - \max\{\tau, \beta(x) \}\big)h(x)g(\beta(x)) + (\beta(x) - \tau)^+ \Big) \d x 
    + \int_0^\tau \Bigg( ( 1 - (\theta(y) - \gamma)^+ )h(y)g(\theta(y)) \nonumber \\
    & \qquad\qquad\qquad  + (\theta(y) - \gamma)^+ + \int_0^{\bar{\theta}(y)} h(x)g(y)\d x + \int_{\bar{\theta}(y)}^\gamma h(x)g(\beta(x))\d x\Bigg) \d y \label{equation:analytical_lower_bound_2+3} .
\end{align}

We show that the above lower bound is minimized when $\beta(x) = \tau$ for all $x < \gamma$.
Clearly, when $\beta(x) \leq \tau$ (in which case $\max\{\tau, \beta(x) \} = \tau$ and $(\beta(x) - \tau)^+ = 0$), the lower bound is non-increasing in $\beta(x)$ since $g$ is non-increasing.
Now consider the partial derivative of the lower bound over $\beta(x) > \tau$, we obtain (for ease of notation, we use $\beta$ to denote $\beta(x)$):
\begin{align*}
    & \int_0^\gamma \big( 1 - h(x)\cdot g(\beta) + (1-\beta)\cdot h(x)\cdot g'(\beta) \big) \d x + \int_0^\tau \int_{\bar{\theta}(y)}^\gamma h(x)\cdot g'(\beta) \d x \d y \\
    \geq &\ \int_0^\gamma \big( 1 - h(x)\cdot g(\beta) + (1-\beta)\cdot h(x)\cdot g'(\beta) \big) \d x + \tau\cdot \int_{0}^\gamma h(x)\cdot g'(\beta) \d x  \\
    \geq &\ \int_0^\gamma \big( 1 - h(x)\cdot g(\beta) + h(x)\cdot g'(\beta) \big) \d x \geq 0,
\end{align*}
where the first inequality holds since $g'(\beta) \leq 0$; the second inequality holds since $\beta > \tau$, the last inequality holds since $h(x)\cdot (g(\beta) - g'(\beta)) \leq 1$, by the assumption of Lemma~\ref{lemma:analytical_lower_bound}.

Therefore the RHS of Equation~\eqref{equation:analytical_lower_bound_2+3} is minimized when $\beta(x) = \tau$ for all $x$, which becomes
\begin{align}
    & \int_0^\gamma (1 - \tau)h(x)g(\tau) \d x 
    + \int_0^\tau \Bigg( ( 1 - (\theta(y) - \gamma)^+ )h(y)g(\theta(y)) \nonumber \\
    & \qquad\qquad\qquad  + (\theta(y) - \gamma)^+ + \int_0^{\bar{\theta}(y)} h(x)g(y)\d x + \int_{\bar{\theta}(y)}^\gamma h(x)g(\tau)\d x\Bigg) \d y \label{equation:analytical_lower_bound_2+3_simplified} .
\end{align}

Notice that the above lower bound (together with Equation~\eqref{equation:analytical_lower_bound_1}) is already very similar to the one claimed in Lemma~\ref{lemma:analytical_lower_bound}.
It remains to show that the above lower bound attains its minimum when $\theta(y) \leq \gamma$.
Recall that $\bar{\theta}(y) = \min\{ \theta(y), \gamma \} = \gamma$ when $\theta(y) > \gamma$.
Hence when $\theta(y) > \gamma$, the partial derivative of Expression~\eqref{equation:analytical_lower_bound_2+3_simplified} over $\theta(y)$ is (where $\theta$ denotes $\theta(y)$)
\begin{equation*}
    \int_0^\tau \big( 1 - h(y)\cdot g(\theta) + (1-\theta+\gamma)\cdot h(y)\cdot g'(\theta) \big) \d y \geq \int_0^\tau \big( 1 - h(y)\cdot (g(\theta) -g'(\theta)) \big) \d y \geq 0.
\end{equation*}

Thus Expression~\eqref{equation:analytical_lower_bound_2+3_simplified} is minimized when $\theta(y) \leq \gamma$, in which case we further lower bound it by
\begin{equation*}
   \int_0^\gamma (1 - \tau)h(x)g(\tau) \d x 
    + \int_0^\gamma \min_{t\leq \gamma} \left\{ h(y)g(t)
   + \int_0^{t} h(x)g(y)\d x + \int_{t}^\gamma h(x)g(\tau)\d x \right\}  \d y .
\end{equation*}

Putting the above lower bound and Equation~\eqref{equation:analytical_lower_bound_1} together completes the proof of the lemma.

\subsection{Fixing \texorpdfstring{$g$}{} and \texorpdfstring{$h$}{}}

Let $g(y) = a - b \cdot \exp(\min\{y,c\})$, where $a = 1.171, b =0.339 , c = 0.652$, and $h(y) = \sqrt{1-g^2(y)}$.
We plot $g$ and $h$ in \Cref{fig:fixing-g-and-h-analysis}.
By definition, $g(y)$ is strictly decreasing when $y\in (0,c)$, and is equal to $g(1) \approx 0.5203$ for all $y\in [c,1)$.
We have $g(y) \in (0.52, 0.832]$ and $h(y)\in (0.5547,0.854)$.
Furthermore, we have
\begin{equation*}
    g'(y) = \begin{cases}
        - b\cdot e^y, & \text{ for } y \leq c \\
        0,  & \text{ for } y > c.
    \end{cases}
\end{equation*}

For notational convenience, we let $G(y) = \int_0^y g(x) \d x$ and $H(y) = \int_0^y h(x) \d x$.
Note that we have $G(1) \approx 0.6329$ and $H(1) \approx 0.76016$.
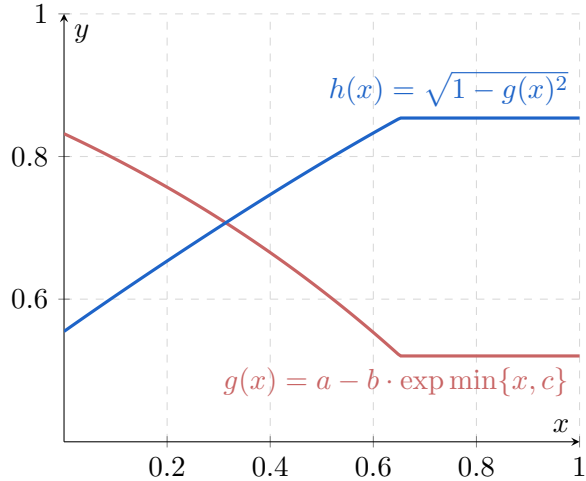
\begin{figure}[!ht]
    \centering
    \begin{tikzpicture}
      \begin{axis}[
        axis lines = middle,
        xlabel = \( x \),
        ylabel = \( y \),
        xmin = 0,    %
        xmax = 1,     %
        ymin = 0.4,    %
        ymax = 1,   %
        grid = both,  %
        grid style = {dashed, gray!30}, %
        samples = 200, %
        domain = 0:1, %
      ]
    
    \addplot[myRed, very thick] {1.171 - 0.339 * exp(min(x, 0.652))} node[below left] {$g(x) = a - b\cdot \exp{\min\{x, c\}}$};
    \addplot[myBlue, very thick] {sqrt(1 - (1.171 - 0.339 * exp(min(x, 0.652)))^2)} node[above left] {$h(x) = \sqrt{1 - g(x)^2}$};
    
    \end{axis}
    \end{tikzpicture}

    \caption{Illustration for $g$ and $h$}
    \label{fig:fixing-g-and-h-analysis}
\end{figure}

We first show that the fixed functions satisfy the condition of Lemma~\ref{lemma:analytical_lower_bound}.

\begin{lemma}    \label{lemma:g-h-sat-cond}
    We have $h(1)\cdot (g(y) - g'(y)) \leq 1$ for all $y\in [0,1)$.
\end{lemma}
\begin{proof}
    By definition for all $y\in [0,1)$ we have
    \begin{align*}
        g(y) - g'(y) \leq a - b\cdot e^y + b\cdot e^y = a,
    \end{align*}
    which implies
    \begin{equation*}
        h(1)\cdot ( g(y) - g'(y) ) \leq h(1)\cdot a = a\cdot \sqrt{1-(a-b\cdot e^c)^2} \approx 0.999992 < 1. \qedhere
    \end{equation*}
\end{proof}

Recall from Lemma~\ref{lemma:analytical_lower_bound} that we have the following lower bound on $\frac{\E[\alpha_u + \alpha_v]}{w_{uv}}$:
\begin{align*}
    (1 - \tau)(1 - \gamma) 
    + (1-\tau)\cdot g(\tau)\cdot H(\gamma) 
    + \int_0^\tau \min_{t \leq \gamma}\{ f(\tau,\gamma,y,t) \} \d y.
\end{align*}
where
\begin{equation*}
    f(\tau,\gamma,y,t) := h(y)\cdot g(t) + (g(y) - g(\tau))\cdot H(t) + g(\tau)\cdot H(\gamma).
\end{equation*}

In the following, we show that this lower bound is at least $0.63245$ for every $\tau$ and $\gamma$.
We first show a lemma that enables us to remove the variable $t$.

\begin{lemma} \label{lemma:minimum_of_theta}
    For all $y\leq \tau$, we have $\frac{\partial f(\tau,\gamma,y,t)}{\partial t} \leq 0$ when $t \leq c$; $\frac{\partial f(\tau,\gamma,y,t)}{\partial t} \geq 0$ when $t > c$.
\end{lemma}
\begin{proof}
    For ease of notation, let $\phi(t) = f(\tau,\gamma,y,t)$.
    For $t > c$, we have $g'(t) = 0$, which implies
    \begin{align*}
        \phi'(t) = (g(y) - g(\tau))h(t) \ge 0.
    \end{align*}

    For $t \leq c$, we have
    \begin{align*}
        \phi''(t) &= h(y)g''(t) + (g(y) - g(\tau))h'(t)\\
        &= h(y)\left(-b\cdot e^t\right) + (g(y) - g(\tau))\frac{b\cdot e^t\cdot g(t)}{\sqrt{1 - g(t)^2}} \\
        & = b\cdot e^t\left( \frac{g(y) - g(\tau)}{\sqrt{1/g(t)^2 - 1}} - h(y) \right) \\
        & \le b\cdot e^t\left( \frac{g(0) - g(1)}{\sqrt{1/g(0)^2 - 1}} - h(0) \right) < b\cdot e^t \cdot (-0.087) < 0.
    \end{align*}

    Therefore we have $\phi'(t) \le \phi'(0) \le h(0)g'(0) + (g(0) - g(1))h(0) < -0.015 < 0$.
\end{proof}

Given this lemma, we have
\begin{equation*}
    \min_{t\leq \gamma}\{ f(\tau,\gamma,y,t) \} = f(\tau,\gamma,y, \min\{ \gamma,c \} ).
\end{equation*}

In the following, we consider two cases: $\gamma \leq c$ and $\gamma > c$, and show that $\frac{\E[\alpha_u + \alpha_v]}{w_{uv}}$ is at least $0.63245$.
We first show a technical lemma, which will be useful for deriving the lower bounds.

\begin{lemma}\label{lemma:g'/h-decreasing}
    When $y\in[0, c]$, $\frac{g'(y)}{h(y)}$ is non-increasing.
\end{lemma}
\begin{proof}
    By definition and that $g(y) \leq g(0) = a-b < 1/a$, we have
    \begin{align*}
        \left( \frac{g'(y)}{h(y)} \right)' 
        & = \frac{b\cdot e^y}{h^2(y)}\left( h'(y) - h(y) \right)\\
        & = \frac{b\cdot e^y}{h^2(y) \sqrt{1 - g(y)^2}}\Big( g(y)\left( g(y) - g'(y) \right) - 1 \Big)\\
        & = \frac{b\cdot e^y}{h^3(y)}\Big( g(y)\cdot a -1 \Big) < 0. \qedhere
    \end{align*}
\end{proof}

\subsection{Lower Bounding the Competitive Ratio}

We first consider the case when $\gamma \leq c$.

\begin{lemma} \label{lemma:analytical_lower_bound_gamma_le_c}
    When $\gamma \leq c$, we have $\frac{\E[\alpha_u + \alpha_v]}{w_{uv}} \ge 0.63245$.
\end{lemma}
\begin{proof}
    Given the above analysis, we have the following lower bound on $\frac{\E[\alpha_u + \alpha_v]}{w_{uv}}$:
    \begin{align*}
    &\ (1 - \tau)(1 - \gamma) 
    + (1-\tau)\cdot g(\tau)\cdot H(\gamma) 
    + \int_0^\tau f(\tau,\gamma,y,\gamma) \d y \\
    = &\ (1 - \tau)(1 - \gamma) + \left((1-\tau)\cdot g(\tau) + G(\tau) \right) H(\gamma) + H(\tau)\cdot g(\gamma) \\
        \geq & (1 - \tau)(1 - \gamma) + G(1)\cdot H(\gamma) +  H(\tau)\cdot g(\gamma) = \ell(\tau, \gamma),
    \end{align*}
    where the inequality holds since $g(y)$ is non-increasing and $(1 - \tau)g(\tau) \ge \int_{\tau}^{1}g(y)\d y$.

    When $\tau \leq c$, we have (recall that we also have $\gamma \leq c$)
    \begin{align*}
        \frac{\partial \ell(\tau, \gamma)}{\partial \gamma} & = \tau - 1 + G(1)\cdot h(\gamma) + H(\tau)\cdot g'(\gamma) \\
        & = \tau - 1 + h(\gamma)\cdot \left(G(1) + \frac{g'(\gamma)}{h(\gamma)}H(\tau) \right)\\
        & \le c - 1 + h(c)\cdot \left(G(1) + \frac{g'(0)}{h(0)}H(c) \right) < -0.049<0,
    \end{align*}
    where the first inequality holds by Lemma~\ref{lemma:g'/h-decreasing}.
    Therefore we have 
    \begin{equation*}
        \ell(\tau, \gamma) \geq \ell(\tau, c) = 
        (1 - \tau)(1 - c) + G(1)\cdot H(c) +  H(\tau)\cdot g(c).
    \end{equation*}

    Observe that $\frac{\partial \ell(\tau, c)}{\partial \tau} = h(\tau)\cdot g(c) - (1-c)$ is non-decreasing in $\tau$. Let $h(\tau^*) = \frac{1-c}{g(c)}$, we have
    \begin{equation*}
        \tau^* = \ln\left( \frac{1}{b}\cdot a - \sqrt{1-\left( \frac{1-c}{a-b\cdot e^c} \right)^2} \right) \approx 0.2321,
    \end{equation*}
    which implies
    \begin{equation*}
        \ell(\tau,c) \geq \ell(\tau^*,c) = (1 - \tau^*)(1 - c) + G(1)\cdot H(c) +  H(\tau^*)\cdot g(c) > 0.634.
    \end{equation*}

    Next we consider the case when $\tau \ge c$. 
    Since we have $h(y) = h(c)$ for all $y > c$, the function $\ell(\tau, \gamma)$ is linear in $\tau$. 
    Hence, $\ell(\tau, \gamma) \ge \min\{\ell(c, \gamma), \ell(1, \gamma)  \}$, where 
    \begin{align*}
        \ell(c, \gamma) \ge \ell(c, 0) = (1 - c) + g(0)\cdot H(c) > 0.73,
    \end{align*}
    where the first inequality holds since $\frac{\partial \ell(c, \gamma)}{\partial \gamma} < 0$ by our previous analysis.
    Finally, observe that
    \begin{align*}
        \ell(1, \gamma) = G(1)\cdot H(\gamma) + H(1)\cdot g(\gamma),
    \end{align*}
    whose derivative over $\gamma$ is
    \begin{align*}
        \frac{\partial \ell(1, \gamma)}{\partial \gamma} = G(1)\cdot h(\gamma) + H(1)\cdot g'(\gamma) 
        \ge \left( G(1)\cdot \frac{h(0)}{g'(0)} + H(1) \right) \cdot g'(\gamma) > -0.27\cdot g'(\gamma) > 0,
    \end{align*}
    the first inequality holds by Lemma~\ref{lemma:g'/h-decreasing}.
    Therefore we have
    \begin{equation*}
        \ell(1, \gamma) \ge \ell(1, 0) = g(0)\cdot H(1) > 0.63245. \qedhere
    \end{equation*}
\end{proof}

It remains to consider the case when $\gamma > c$.

\begin{lemma} \label{lemma:analytical_lower_bound_gamma_ge_c}
    When $\gamma > c$, we have $\frac{\E[\alpha_u + \alpha_v]}{w_{uv}} \ge 0.63245$.
\end{lemma}
\begin{proof}
    Recall that in this case we have $\min_{t\leq \gamma}\{ f(\tau,\gamma,y,t) \} = f(\tau,\gamma,y, c )$, and thus
    \begin{align*}
        \frac{\E[\alpha_u + \alpha_v]}{w_{uv}}
        \geq (1 - \tau)(1 - \gamma) + g(\tau)\cdot H(\gamma) + g(c)\cdot H(\tau) + (G(\tau) - \tau\cdot g(\tau))\cdot H(c).
    \end{align*}

    Since $g(y) = g(1)$ and $h(y) = h(1)$ for all $y\geq c$, the lower bound is linear in $\gamma \in (c,1]$, and achieves its minimum either when $\gamma = c$ or $\gamma = 1$.
    The case when $\gamma = c$ has been considered in the previous analysis, where a lower bound of $0.63245$ is derived.
    It remains to consider the case when $\gamma = 1$.
    Let $\phi(\tau)$ be this lower bound:
    \begin{equation*}
        \phi(\tau) = g(\tau)\cdot H(1) + g(c)\cdot H(\tau) + (G(\tau) - \tau\cdot g(\tau))\cdot H(c).
    \end{equation*}

    We verify that $\phi'(\tau) > 0$ for all $\tau$:
    \begin{itemize}%
        \item When $\tau \ge c$, we have $\phi'(\tau) = g(c)\cdot h(\tau) > 0$.
        \item When $\tau \le c$, we have
        \begin{align*}
            \phi'(\tau) & = g'(\tau)\cdot (H(1) - \tau\cdot H(c)) + g(c)\cdot h(\tau) \\
            & = h(\tau)\cdot \left(g(c) + \frac{g'(\tau)}{h(\tau)}\cdot ( H(1) - \tau\cdot H(c) ) \right) \\
            & \ge h(\tau)\cdot \left(g(c) + \frac{g'(c)}{h(c)}\cdot ( H(1) - c \cdot H(c) ) \right) \approx 0.1711\cdot h(\tau) > 0,
        \end{align*}
        where the first inequality holds since $\frac{g'(\tau)}{h(\tau)}\cdot (H(1) - \tau\cdot H(c))$ is non-increasing in $\tau$, which can be verified taking a derivative over $\tau$ (similar to the proof of Lemma~\ref{lemma:g'/h-decreasing}).
    \end{itemize}

    Therefore we have $\phi(\tau) \ge \phi(0) = g(0)\cdot H(1) > 0.63245$.
\end{proof}

\section{Redundancy Reduction}\label{app:redundancy-reduction}

In this section we show that when solving the QCQP (for a fixed $n$), we do not need to enumerate all possible step functions $\theta, \beta\in \mathcal S_n$, by identifying and removing some redundant constraints.

\begin{lemma}\label{app:lemma:redundancy-reduction}
    In the implementation of computer-aided approaches, it is w.l.o.g. to only consider $\theta, \beta\in \mathcal S_n$ such that $\sum_{i = 1}^n \Theta_i \ge \sum_{i = 1}^n B_i$ and $\forall i\in[n] : \Theta_i \ge {B}^{-1}_i$.
\end{lemma}
\begin{proof}
    For $n$-segment step functions $g$ and $h$, we define 
    \begin{align*}
        F(\theta, \beta) = \sum_{i = 1}^n\big(\Theta_i - B_i^{-1}\big)^+ & + \sum_{i = 1}^{n} \Big( \big(1 - (\Theta_i - B_i^{-1})^+\big) \cdot H_i \cdot G_{n\Theta_i + 1} \Big) \\
        & + \sum_{i = 1}^n \Big( \big(1 - (B_i - \Theta_i^{-1})^+ \big) \cdot H_i \cdot G_{nB_i + 1} \Big).
    \end{align*}

    Observe that (see \Cref{fig:symmetry-theta-beta} for an example)
    \begin{equation*}
        \sum_{i = 1}^n\big(\Theta_i - B_i^{-1}\big)^+ 
        = \sum_{i=1}^n \sum_{j=1}^n \left( \mathbf{1}\left( \frac{j}{n}\leq \Theta_i, \frac{i}{n}\leq B_j \right) \right).
    \end{equation*}

    Therefore $F(\theta,\beta)$ is symmetric on $\theta$ and $\beta$, i.e., $F(\theta,\beta) = F(\beta,\theta)$, and the first statement follows.
    Now we proof the second statement by showing that for all $\theta,\beta\in \mathcal S_n$, we have
    \begin{equation*}
        F(\hat{\theta}, \beta) \leq F(\theta, \beta), \quad \text{ where }
        \hat{\theta} = \max\{ \theta, \beta^{-1} \}.
    \end{equation*}

    Note that $\hat{\theta}$ is also a non-decreasing function in $\mathcal{S}_n$. Moreover, by our definition of inverse functions, we have $\hat{\Theta}^{-1}_i \leq {\Theta}^{-1}_i$ for all $i$.
    Clearly, replacing $\theta$ with $\hat{\theta}$ does not change the value of the first term of $F(\theta,\beta)$, because
    \begin{equation*}
        (\hat{\Theta}_i - B^{-1}_i)^+ = (\max\{{\Theta}_i, B^{-1}_i\} - B^{-1}_i)^+ = (\Theta_i - B^{-1}_i)^+.
    \end{equation*}
    
    The replacement does not increase the second term because the coefficient of each summand is not changed and $g$ is non-increasing.
    For the third term, notice that for each $i$ we have
    \begin{equation*}
        \big(1 - (B_i - \hat{\Theta}_i^{-1})^+ \big) \cdot H_i \cdot G_{nB_i + 1} 
        \leq  \big(1 - (B_i - \Theta_i^{-1})^+ \big) \cdot H_i \cdot G_{nB_i + 1},
    \end{equation*}
    since $\hat{\Theta}^{-1}_i \leq {\Theta}^{-1}_i$.
    Therefore, the replacement does not increase the third term either, which implies $F(\hat{\theta}, \beta) \leq F(\theta, \beta)$ and concludes the proof.
\end{proof}

\begin{figure}[ht]
    \centering
    \begin{tikzpicture}
        \draw[gray!20, thin] (0,0) grid (5,5);
        \draw[thick, -latex] (0,0) -- (5.5,0) node[below=0.15] {$y_v$};
        \draw[thick, -latex] (0,0) -- (0,5.5) node[left] {$y_u$};
        \node[black] at (-0.25, -0.25) {$O$};
    
        \foreach \x in {1,...,5}
            \draw (\x,0.1) -- (\x,-0.1) node[below] {$\pgfmathparse{\x/5}\pgfmathprintnumber[fixed, precision=1]{\pgfmathresult}$};
        \foreach \y in {1,...,5}
                \draw (0.1,\y) -- (-0.1,\y) node[left] {$\pgfmathparse{\y/5}\pgfmathprintnumber[fixed, precision=1]{\pgfmathresult}$};
    
        \draw[draw=none,fill=gray!20] (1,1) rectangle (2,2);
        \draw[draw=none,fill=gray!20] (3,3) rectangle (5,4);
        \draw[draw=none,fill=gray!20] (4,4) rectangle (5,5);
                
        \draw[myRed, very thick] (0,0) -- (1,0);
        \draw[myRed, very thick, dashed] (1,0) -- (1,2); 
        \draw[myRed, very thick] (1,2) -- (3,2); 
        \draw[myRed, very thick, dashed] (3,2) -- (3,4); 
        \draw[myRed, very thick] (3,4)-- (4,4); 
        \draw[myRed, very thick, dashed] (4,4) -- (4,5);
        \draw[myRed, very thick] (4,5)-- (5,5); 

        \node[draw=myRed, fill=myRed, circle, minimum size=3pt, inner sep=0pt] at (0, 0) {};
        \node[draw=myRed, fill=myRed, circle, minimum size=3pt, inner sep=0pt] at (1, 2) {}; 
        \node[draw=myRed, fill=myRed, circle, minimum size=3pt, inner sep=0pt] at (3, 4) {};
        \node[draw=myRed, fill=myRed, circle, minimum size=3pt, inner sep=0pt] at (4, 5) {};
        \node[draw=myRed, fill=white, circle, minimum size=3pt, inner sep=0pt] at (1, 0) {};
        \node[draw=myRed, fill=white, circle, minimum size=3pt, inner sep=0pt] at (3, 2) {};
        \node[draw=myRed, fill=white, circle, minimum size=3pt, inner sep=0pt] at (4, 4) {};
        \node at (4.5,5.5) {\textcolor{myRed}{$\theta(y_v)$}};
        \draw[myBlue, very thick] (0, 0) -- (0,1); 
        \draw[myBlue, very thick, dashed] (0,1) -- (2,1);
        \draw[myBlue, very thick] (2,1) -- (2,3); 
        \draw[myBlue, very thick, dashed] (2,3) -- (5,3); 
        \draw[myBlue, very thick] (5,3) -- (5,5);
        \node[draw=myBlue, fill=white, circle, minimum size=3pt, inner sep=0pt] at (0,1) {};
        \node[draw=myBlue, fill=white, circle, minimum size=3pt, inner sep=0pt] at (2,3) {};
        \node[draw=black, fill=black, circle, minimum size=3pt, inner sep=0pt] at (5,5) {};
        \node[draw=black, fill=black, circle, minimum size=3pt, inner sep=0pt] at (0,0) {};
        \node[draw=myBlue, fill=myBlue, circle, minimum size=3pt, inner sep=0pt] at (2,1) {};
        \node[draw=myBlue, fill=myBlue, circle, minimum size=3pt, inner sep=0pt] at (5,3) {};
        \node at (5.8,4) {\textcolor{myBlue}{$\beta(y_u)$}};
    \end{tikzpicture}
    \caption{Demonstration for $\sum_{i = 1}^n\big(B_i - \Theta_i^{-1}\big)^+ = \sum_{i = 1}^n\big(\Theta_i - B_i^{-1}\big)^+$.}
    \label{fig:symmetry-theta-beta}
\end{figure}
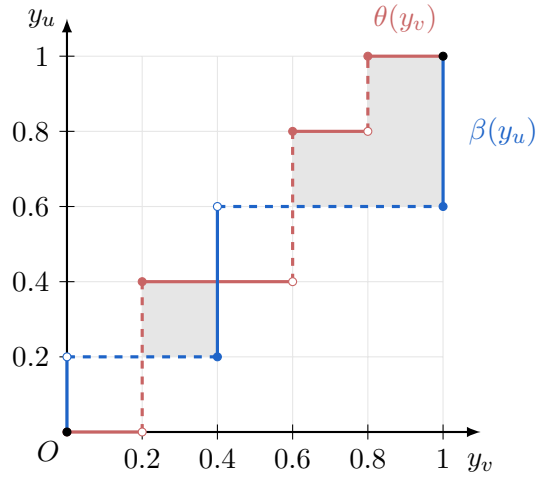

\end{document}